\documentclass{article}
\usepackage[utf8]{inputenc}
\usepackage[margin=1.2in,footskip=0.25in]{geometry}
\usepackage{amsmath, amsthm, amscd, amsfonts, amssymb, graphicx, color}
\usepackage[bookmarksnumbered, colorlinks, plainpages]{hyperref}
\usepackage{mathtools,amsmath}
\usepackage[title]{appendix}
\usepackage{commath}
\usepackage{graphicx}
\usepackage{enumerate}
\usepackage{amssymb}
\usepackage{amsthm}
\usepackage{dsfont}
\usepackage{hyperref}
\usepackage{tcolorbox}
\usepackage{subcaption}
\usepackage{mdframed}
\usepackage{todonotes}
\usepackage{dirtytalk}
\hypersetup{
  pdfpagemode=UseNone,
  colorlinks=true,
  citecolor=blue,
  linkcolor=blue,
  urlcolor=blue
}
\usepackage{thmtools}
\usepackage{thm-restate}
\newcommand\newsubcap[1]{\phantomcaption%
       \caption*{\figurename~\thefigure(\thesubfigure): #1}}

\declaretheorem[name=Theorem]{thm}
\declaretheorem[name=Lemma,sibling=thm]{lemma}

\declaretheorem[name=Claim,sibling=thm]{claim}

\newcommand{\I}{\mathds{1}}
\newcommand{\e}{\mathrm{e}}
\newcommand{\ii}{\mathrm{i}}

\newcommand{\sinc}{\mathrm{sinc}}

\DeclarePairedDelimiter\bra{\langle}{\rvert}
\DeclarePairedDelimiter\ket{\lvert}{\rangle}
\newcommand{\braket}[2]{ \langle #1 | #2 \rangle}
\newcommand{\ketbra}[2]{ | #1 \rangle \! \langle #2 |}

\definecolor{ha}{RGB}{246,76,246}

\usepackage{multirow}

\usepackage{hyperref}
\usepackage[capitalise]{cleveref}
\crefname{section}{Section}{Sections}
\crefname{subsection}{subsection}{subsections}
\crefname{thm}{Theorem}{Theorems}
\crefname{prop}{Proposition}{Propositions}
\crefname{claim}{Claim}{Claims}
\crefname{corollary}{Corollary}{Corollaries}
\crefname{lemma}{Lemma}{Lemmas}
\crefname{appendix}{Appendix}{Appendices}
\crefname{defn}{Definition}{Definitions}
\crefname{equation}{Eq.}{Eqs.}
\crefname{alg}{Algorithm}{Algorithms}

\title{Reducing quantum resources for observable estimation with window-assisted coherent QPE}
\author{Harriet Apel$^{1,2}$, Cristian L. Cortes$^{1}$, Jessica Lemieux$^{1}$, Mark Steudtner$^{1}$}
\date{\small$^1$\textit{PsiQuantum, 700 Hansen Way, Palo Alto, CA 94304, USA} $^2$\textit{Department of Computer Science, UCL, Gower Street, London, UK}\normalsize}

\usepackage{natbib}
\usepackage{graphicx}
\usepackage{array,multirow}
\usepackage{makecell}
\usepackage{siunitx}
\usepackage{booktabs}
\usepackage{xcolor}

\begin{document}

\maketitle

\begin{abstract}
Quantum Phase Estimation (QPE) routines are known to fail probabilistically even with perfect gates and input states.
This effect stems from an incompatibility of finite-sized quantum registers to capture a phase within QPE with phase angles of infinite precision, and the effect extend even beyond what would be reasonably expected from rounding. 
This effect can be partially mitigated by biasing the phase register with a window, or taper state, from classical signal processing. 
This paper focuses on how windowing a \emph{coherent} QPE used as a subroutine can improve the accuracy of the overall algorithm.
Specifically we study the quantum task of estimating observables where window-assisted coherent QPE is used as a subroutine to implement a reflection about an eigenstate.
Quantum resource estimates show over 2-orders-of-magnitude reduction in Toffoli counts over the previous costed techniques -- also assisted by the use of improved block encoding techniques -- demonstrating an encouraging decrease in resources for quantum computation of molecular observables.
Since QPE, as one of only a few quantum building blocks, appears as a subroutine in many algorithms; this analysis also provides a model for understanding how window functions propagate to an improved error in composite algorithms.
\end{abstract}

Quantum phase estimation (QPE) is a seminal routine in quantum algorithms, leveraging the quantum Fourier transform (QFT) to learn the eigenphase of a unitary with respect to a particular eigenstate.
In addition to the qubits required to describe the system, the routine requires a phase register that is measured to extract the eigenphase. 
In practice, the phase register consists of a finite number of qubits and introduces a `bit discretisation' effect.
This does not only result in the coarsening of data into bins of finite precision, but distorts the probability distribution of the measured phase state. 
Therefore, even given an exact eigenstate as input, with non-zero probability the final phase register measurement obtains an outcome far from the true eigenphase -- this failure probability can be as high as $19\%$.
This affects the success probability of the routine, necessitating repetitions to learn the phase with confidence.
Bit discretisation can be well understood from the vantage of classical signal processing, as this distortion has the same origins as performing a Fourier transform on a time-limited signal \cite{Sanders2020}.

It was also noted that preparing the register that holds the phase estimate in a superposition with non-uniform amplitudes can reduce the effect of bit discretisation.
A uniform superposition is akin to sampling the signal from a rectangular window with sharp discontinuities at the periodic boundary, whereas a tapered window function dampens these  discontinuity and reduces the distortion of the probability distribution.
While being ubiquitous in classical signal processing, window functions have been used in quantum both to asymptotically reduce query counts \cite{Berry, Coupledoscillator23} and to practically outperform quantum signal processing (QSP) in small instances \cite{greenaway2024case}.

There are two important quantities to assess the utility of an algorithm: the \emph{success probability} describing the likelihood that a given run of the algorithm produces a \say{correct result}; and the \emph{error} describing the maximum deviation from the true value that we are willing to accept as a \say{correct result}.
For example, if the true value is $x_0$, we allow an error of $\epsilon$ and the algorithm produces the result $x$ according to some distribution, then the success probability of the algorithm is $P(\, \abs{x-x_0}\leq \epsilon)$.
In standalone QPE -- where the phase register is directly measured -- window functions can increase the success probability of the algorithm for a given error but increasing the accuracy can only be improved by increasing the number of phase qubits to achieve a lower error.
This work studies window functions for \emph{coherent} QPE subroutines within a larger algorithm.
The coherent nature of the QPE indicates that the phase register is not measured at the end of the routine and instead further gates are applied using the full phase register state for computation. 
Therefore, instead of a failure probability, we consider a failure amplitude of the subspace spanned by $\{\ket{x}\}$ with $\abs{x-x_0}>\epsilon$.
This introduces a mid-circuit error, hence improving the success amplitude using window functions in coherent QPE increases the accuracy of the overall algorithm and can help achieve a lower error tolerance. 
Improving the algorithmic error is more compelling than improving the success probability, since while the latter can be compensated by taking repeated measurements, improving the former requires resources such as additional gates and logical qubits which will determine whether such algorithms can be run on limited hardware at all.

A notable application of coherent QPE is in estimating the expectation value of an operator with respect to a system's ground state \cite{Knill2006,Obrien2020,seve}.
This is an interesting case study since the concatenated nature of the algorithm nicely demonstrates the propagation from the QPE success amplitude to the overall algorithmic error.
Furthermore, it is one of the few fault-tolerant algorithms using a coherent QPE where the quantum resources have been numerically estimated \cite{seve}.
The natural context for this algorithm is the estimation of quantum observables on ground states of molecular systems, which has profound applications in computational chemistry.
Expectation values of unperturbed wavefunctions may be used to estimate first-order molecular properties such as molecular multipole moments, forces on nuclei and stationary points on potential energy surfaces \cite{frenkelʹ1934wave, Clusius1941EinfhrungID}. Furthermore, they can also be used to compute higher-order properties such as polarizability and harmonic vibrational frequencies with additional post-processing \cite{Helgaker}. These properties play a significant role in predicting reactivity and as well as understanding spectroscopic data. The finite phase register in the QPE subroutine introduces a bit discretisation that degrades the algorithmic error -- reducing the precision to which we can learn these molecular properties. 

This work both analyses the error of this algorithm when the QPE phase register is unitarily prepared in a general state and demonstrates the practical benefit of choosing the Kaiser window \cite{Kaiser1980} as the phase register taper function.
The error analysis demonstrates that coherent use of QPE is exactly where introducing window functions can have the most significant reduction in quantum resources. 
As is somewhat expected due to the complex nature of the question, our resource estimates still clearly place this algorithm in the regime requiring a large-scale fault-tolerant quantum computer.
However, the advantage of using window functions over other techniques to combat bit discretisation in QPE (such as applying QSP to round phase values in QPE \cite{Rall2021fastercoherent}) is seen via order-of-magnitude reductions in the number of Toffoli gates required to estimate molecular observables.
We also show that this reduction is further improved by using state of the art block-encoding methods \cite{Rocca2024,loaiza2023block, Zak2024} as well as adapting recent tensor factorisation methods (BLISS-THC) to simultaneously optimise  the Hamiltonian and observable rank truncation parameters.

The technical background of both window-assisted QPE and the expectation value algorithm is first given in \cref{sect Tech}.
Slight modifications to previous expectation value algorithms \cite{Knill2006,Obrien2020,seve} are required to capture the benefit of window-assisted QPE, so this is presented and analysed in \cref{sect algo analysis}.
Finally the numerical evidence of window functions reducing quantum resources is attested in \cref{sect RE} where the Toffoli and qubit counts are compared with the state of the art. 

\section{Technical background}\label{sect Tech}

\subsection{QPE with window functions}

Given controlled access to a unitary $U$, along with copies of an eigenstate $\ket{\psi}$, the goal of QPE is to estimate the eigenphase $\phi\in[0,1]$: $U\ket{\psi} = \e^{2\pi \ii \phi}\ket{\psi}$. 
More generally, given copies of any state $\ket{\nu} = \sum_i c_i \ket{\psi_i}$ written here in the eigenbasis of $U$, measuring the phase qubits of a QPE probabilistically projects the state onto an eigenstate $\ket{\psi_i}$ and estimates the corresponding eigenphase $\phi_i$.
To introduce window functions, we initially assume exact copies of an eigenstate $\ket{\psi}$ for simplicity.
QPE requires a quantum register for the system that supports the input state $\ket{\psi}$ (denoted \say{sys}), and an additional $n$ qubit quantum register to extract the phase (denoted \say{ph}). 
This phase register is unitarily prepared in some superposition over the computational basis states, $\hat{W}\ket{a}_{\text{ph}} = \sum_{x=0}^{2^n-1} W_{a}(x) \ket{x}_{\text{ph}}$.
Given that the phase register is initially prepared in $\ket{0}_{\text{ph}}$ the total input state for the QPE is $\ket{\psi}_{\text{sys}}\otimes \sum_{x=0}^{2^n-1} W_0(x)\ket{x}_{\text{ph}}.$
The typical choice for this function is the rectangular window $W_0(x) = \sqrt{1/2^n}$, efficiently prepared by applying parallel Hadamard gates $\hat{W} = H^{\otimes n}$.
Controlled $U$'s are then applied to kickback the eigenphase of $\ket{\psi}$ onto the phase register.
Finally, an inverse quantum Fourier transform and then measurement of the phase register allows an estimate of the eigenphase to be obtained.
Given the eigenbasis for the Hamiltonian is $\{\ket{\psi_i}\}$ with corresponding eigenphases $\{\phi_i\}$ the unitary action of QPE is given by
\begin{align}
    \textup{QPE} &= \sum_{i=0}^{2^t-1} \ket{\psi_i}\bra{\psi_i}_{\text{sys}} \otimes \sqrt{\frac{1}{2^n}} \sum_{a,k=0}^{2^n-1}\left(\sum_{x=0}^{2^n-1} W_{a}(x) \e^{2\pi\ii (2^n \phi_i - k)\frac{x}{2^n}} \right)\ket{k}\bra{a}_{\text{ph}}.
\end{align}

Consider the state $\textup{QPE}\ket{\psi_i}_{\text{sys}}\ket{0}_{\text{ph}}$ where the rectangular window is used in the QPE, the squared amplitude of the phase register in the $k^{\text{th}}$ computational basis state is given by
\begin{equation}\label{eqn prob k}
    P_k = \abs{\frac{1}{2^n}\sum_{x=0}^{2^n -1}\e^{2\pi \ii (2^n\phi_i - k )\frac{x}{2^n}}}^2.
\end{equation}
If $\phi_i$ is exactly representable by $n$ bits (i.e. $2^n\phi_i \in \mathbb{Z}^+$) then algorithmically the if measured eigenphase is obtained exactly with certainty, i.e. $P_{k=2^n\phi_i} = 1$.
However, as generally this is not the case ($2^n\phi_i$ is non-integer) bit discretisation is introduced.
If one were to directly measure the phase register the outcomes are
 sampled from a probability distribution, requiring repeated runs of the QPE circuit to learn the histogram.
The number of phase qubits can be increased to improve the error and the success probability of the algorithm.
To learn $\phi_i$ to precision $\epsilon$ with probability $(1-\delta)$ requires $n \in O\left(\log\frac{1}{\epsilon} \right)$ qubits and $O\left(\frac{\log(1/\delta)}{\epsilon} \right)$ calls to controlled $U$ \cite{Mande2023}.

\begin{figure}[h!]
\centering
\includegraphics[trim={0cm 0cm 0cm 0cm},clip,width=0.75\linewidth]{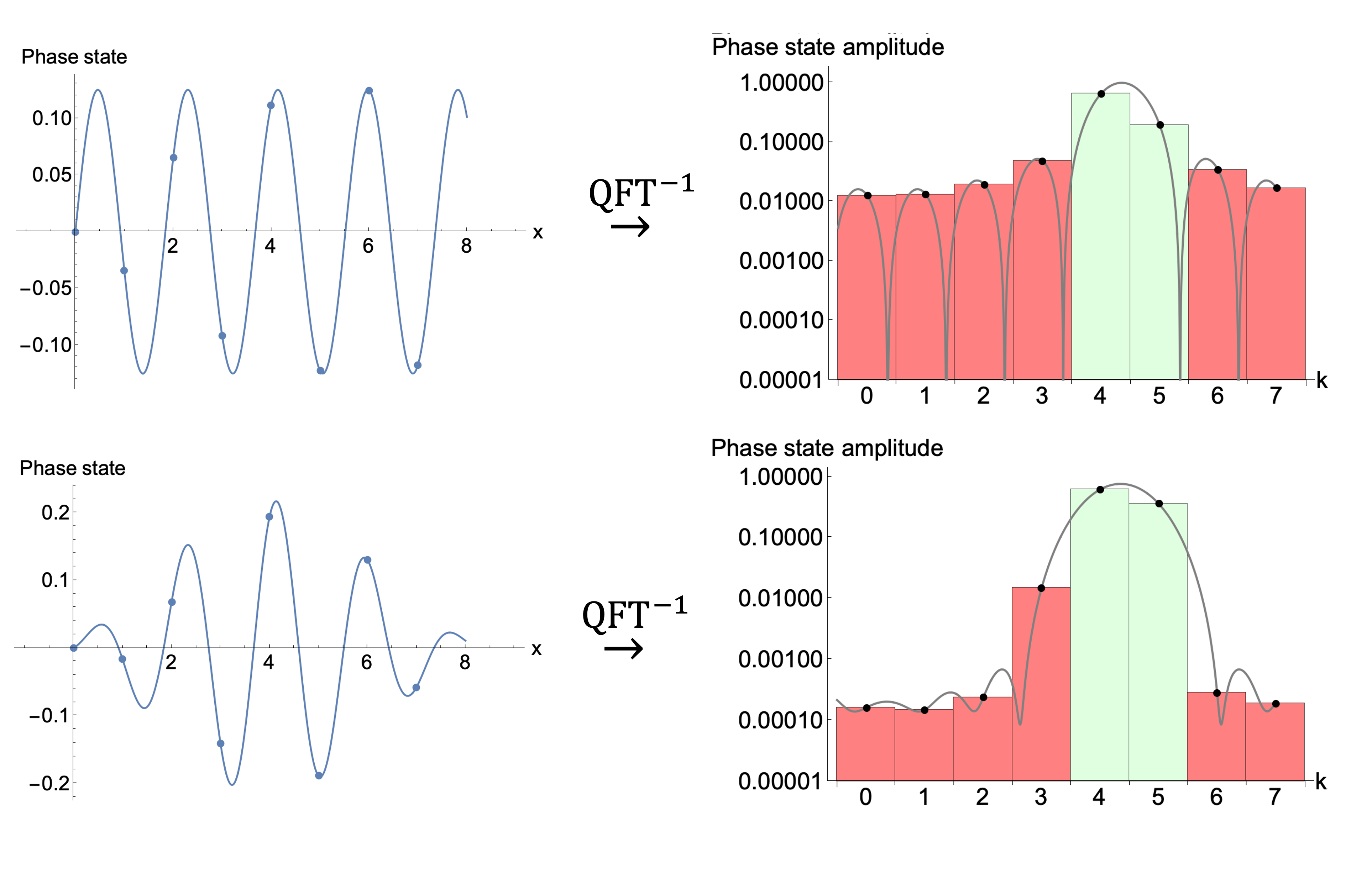}
\caption{Bit discretisation and the link to classical signal processing. \textbf{Top:} An oscillation function is multiplied by a rectangular function, the frequency is such that the signal is not periodic and therefore $2\pi\phi$ is non-integer. After the inverse $\textup{QFT}$ the phase state amplitude distribution is given by a $\sinc$ function centered at $2\pi\phi$. When sampled, due to the \emph{side bands} of the $\sinc$ function, there is some amplitude in bins that are far from the true frequency -- shown in pink. Note that if $2\pi\phi$ was an integer the distortion of the function is hidden by sampling the zeros of the $\sinc$, recovering the original delta function. \textbf{Bottom:} Now instead of a rectangular window a smooth envelop function modulates the oscillation from above. After the inverse $\textup{QFT}$ the amplitude is more concentrated in the two bins closest to the true phase -- shown in green as the \emph{central lobe} -- due to the smaller side bands of window Fourier transform.    }
\label{fg FT of QPE}
\end{figure}

It is insightful to analyse QPE through the lens of Fourier analysis.
Inside the modulus of the right-hand side of \cref{eqn prob k} is the discrete Fourier transform (DFT) of an oscillating function with frequency $\phi$.
The window function modulates the oscillation $\e^{2\pi \ii \phi x}\mapsto W_a(x)\e^{2\pi \ii \phi x}$ before the DFT.
In the resulting probability distribution, the frequency domain of the window function is convoluted with the delta peak $\delta_{k,2\pi\phi}$.
This convolution distorts the distribution giving rise to non-zero amplitude to $\ket{k}$ where $k$ is far from $2\pi\phi$ .
The rectangular window, due to the large sidebands of the $\sinc$ function, has particularly problematic distortion.

This viewpoint of Fourier analysis links the problem of \emph{bit discretisation} in QPE to the problem of \emph{spectral leakage} in classical signal processing, a well-studied problem \cite{carlson2010communication,spectral_leakageHarris}.
One solution from signal processing is to replace the rectangular window function to reduce this effect.
Given a time-limited function, instead of using the discontinuous rectangle function, the envelope is  tapered to `smooth out' the edges.
This is helpful since the FT of a tapered function has less amplitude in side-bands than the sinc function \cite{Kaiser1980} and therefore can lead to less distortion.
See \cref{fg FT of QPE} for an illustration.
This idea can be applied to QPE by preparing the phase register, not in an equal superposition, but with some non-uniform window function with tapered amplitudes \cite{Sanders2020}.
These amplitudes carry through the circuit so the squared amplitude of $\ket{k}_{\text{ph}}$ is now more generally:
\begin{equation}
    P_k = \abs{\sqrt{\frac{1}{2^n}}\sum_{x=0}^{2^n-1}W_0(x) \e^{2\pi\ii (2^n \phi - k)\frac{x}{2^n}}}^2.
\end{equation}

Many different window functions exist, one good candidate is the Kaiser function \cite{Kaiser1980}:
\begin{equation}\label{eqn Kaiser}
    K(x,\beta) = \sqrt{\frac{1}{\mathcal{N}}} \frac{I_0(\beta \sqrt{1-\bar{x}^2})}{I_0(\beta)},
\end{equation}
where $I_0$ is the $0^{\text{th}}$ order modified Bessel function and  $\mathcal{N}$ is a normalisation constant. 
To use this window function the phase register is prepared in the Kaiser state: $\hat{W}\ket{0}_{\text{ph}} = \sum_{x=0}^{2^n-1} K(x,\beta) \ket{x}_{\text{ph}}$.
The Kaiser window is tractable to analysis -- it has been shown to reduce the query complexity of QPE with error $\epsilon$ and success probability $(1-\delta)$ to $O\left( \epsilon^{-1}\log(\delta^{-1})\right)$ \cite{Berry} -- an exponential improvement in $\delta$.
The window is tunable with the bandwidth parameter, $\beta$.
Larger $\beta$ corresponds to a narrower distribution and therefore a Fourier transform with wider bandwidth and lower amplitude in sidebands. 

The exponentially improved complexity of QPE with the Kaiser window has been used to analyse more complex quantum algorithms \cite{Coupledoscillator23,Berry,berry2024rapidinitialstatepreparation}.
However, there has been little work considering the non-asymptotic advantage of using windowed QPE in a coherent setting to combat bit discretisation over other more involved rounding techniques such as QSP. 

\subsubsection*{Coherent QPE for reflections}\label{sect QPE reflection}

QPE can also be a useful subroutine within a more complex algorithm.
Here it is performed coherently and failure amplitude in the phase register is fed through to subsequent gates and contributes to a mid-circuit error. 
Therefore, instead of manifesting as a failure probability, in coherent QPE the finite size effects of the phase register propagates to an algorithmic error. 

QPE can be used coherently to approximately implement a reflection about an eigenstate $R_{\psi_0} = \mathds{1} - 2 \ket{\psi_0}\bra{\psi_0}$ within a subspace, given a $t$-qubit unitary $U$ such that $\ket{\psi_0}$ is an eigenstate with eigenphase $\phi_0$.
QPE acts jointly on this system register of $t$ qubits and an ancillary phase register of $n$ qubits.
To implement an approximate reflection about a state using QPE, the $n$-bit approximation of the eigenphase (denoted $\tilde{\phi}_0$) must be known in advance.
The circuit for the reflection then consists of a QPE of $U$, a phase register reflection about $\ket{\tilde{\phi}_0}_\text{ph}$ and finally uncomputing the QPE: $R_{\psi_0} \approx \textup{QPE}^\dagger (\mathds{1} - 2\ket{\tilde{\phi}_0}\bra{\tilde{\phi}_0}_{\text{ph}})\textup{QPE}$.
This can be less resource intensive than a direct implementation using QSP which requires precisely computing many phase factors~\cite{Gilyen2018,Haah2019,DongMeng}. 

Inaccuracy in phase estimation now leads to implementing a distorted unitary reflecting about some contaminated state.
Instead of reflecting a single eigenstate $\ket{\phi_0}$ the reflection occurs about a perturbed state with non-zero amplitude in other eigenstates. 
Unless \emph{all} eigenphases are exactly representable in $n$ bits, the reflection unitary will be perturbed.
The QPE based reflection takes the form:
\begin{equation}\label{eqn QPE refl}
\begin{multlined}
   \textup{QPE}^\dagger (\mathds{1} - 2\ket{\tilde{\phi}_0}\bra{\tilde{\phi}_0}_{\text{ph}})\textup{QPE} = \mathds{1} - 2\sum_{i=0}^{2^t-1} \ket{\psi_i}\bra{\psi_i}_{\text{sys}} \otimes\\ \frac{1}{2^n} \sum_{a,a'=0}^{2^n-1} \left(\sum_{x=0}^{2^n-1} W_{a}(x) \e^{2\pi\ii ( \phi_i - \tilde{\phi}_0)x} \right)\left(\sum_{x'=0}^{2^n-1} W_{a'}^*(x') \e^{-2\pi\ii ( \phi_i - \tilde{\phi}_0)x'} \right)\ket{a'}\bra{a}_{\text{ph}}.
   \end{multlined}
\end{equation}
Using an appropriate window function can bring this operation closer to the intended reflection about $\ket{\psi_0}_\text{sys}\otimes \ket{0}_\text{ph}$.

\subsection{Expectation value algorithms}

One of the attractive future applications of quantum computers is to improve speed and accuracy of \emph{in silico} quantum chemistry simulations. 
Aside from molecular ground state energy estimation, a task with promising quantum advantage is estimating other molecular observables with respect to a state of interest. 
Algorithmically the goal is to estimate $\bra{\sigma_0}\hat{F}\ket{\sigma_0}$: the expectation value of some observable $\hat{F}$ with respect to the ground state, $\ket{\sigma_0}$, of a Hamiltonian, $H$.
Moreover, generally the Hamiltonian does not commute with $\hat{F}$.

The approach of repeatedly preparing the state $\ket{\sigma_0}$ to measure the expectation value requires $O(\epsilon^{-2})$ queries to a block encoded Hamiltonian.
This was improved in \cite{Knill2006, Rall2020,Obrien2020,seve} to $O(\epsilon^{-1})$ by instead performing amplitude estimation \cite{Brassard2000} of some walk operator\footnote{There is also work considering the complexity of estimating $M$ observables. \cite{Huggins2022} show that the $O(\epsilon^{-1})$ scaling in the worst case comes at a cost of linear complexity in the number of observables. Shadow estimation protocols \cite{Aaronson2018} obtain a $O(\log(M))$ scaling at the price of a poor accuracy scaling $O(\epsilon^{-4})$. This can be improved to $\tilde{O}(\epsilon^{-2})$ by using randomised shadow estimation circuits \cite{Huang2020} but only for restricted observables that do not introduce a scaling of the system Hilbert space into the complexity.}.
This work leverages the walk operator from \cite{Obrien2020,seve} $\mathcal{U}:=(\mathcal{B}[\hat{F}])\cdot (\textup{c}\mbox{-}R_{\sigma_0} )$, the product of $\mathcal{B}[\hat{F}]$ a \emph{self-inverse} block encoding of the non-unitary observable and $\textup{c}\mbox{-}R_{\sigma_0}$ a controlled reflection about $\ket{\sigma_0}$.
The block encoding of the observable $\hat{F}$ necessitates an ancillary quantum register we label $\mathcal{B}_F$. 
The reflection about the state is then controlled on $\ket{0}_{\mathcal{B}_F}$ to select for the encoding subspace.

\begin{restatable}[Block encoding]{defn}{defnBE}\label{defn block encoding}
Consider the Hermitian operator $\hat{F}$ acting on $\mathcal{H}_{\textup{sys}}$.
A self-inverse block encoding of the operator $\mathcal{B}[\hat{F}]$ is a \emph{unitary} operator acting on $\mathcal{H}_{\textup{sys}} \otimes \mathcal{H}_{\mathcal{B}_F}$ that encodes $\hat{F}$ in a subspace $\left(\bra{0}_{\mathcal{B}_F}\otimes \mathds{1}_{\textup{sys}}\right)\mathcal{B}[\hat{F}]\left(\ket{0}_{\mathcal{B}_F}\otimes \mathds{1}_{\textup{sys}}\right) = \hat{F}$.
\end{restatable}

The algorithm hinges on a relation between the eigenphase of the walk operator $\mathcal{U}$ and the desired expectation value $\bra{\sigma_0}\hat{F}\ket{\sigma_0}$.
This can be understood via a general result about the product of two reflections:

\begin{restatable}{lemma}{lmiteratespec}\label{lm iterate spectrum}[\cite{Szegedy} or \cite[Appendix A]{nagaj2009fastamplificationqma}]
Consider two reflections:
\[    R_A  = \mathds{1} - 2 \Pi_A  \qquad \text{and} \qquad
    R_B  = \mathds{1} - 2\Pi_B,\]
where $\omega_k$ are the singular values of the projector product $\Pi_A \cdot \Pi_B$.
The product of these reflections $\mathcal{U}:= R_B R_A$ is a unitary with eigenphases:
\[\theta_{k,\pm} = \mp \frac{1}{2\pi} \cos^{-1}\left(2\abs{\omega_k}^2 -1 \right).\]
\end{restatable}

While a known result, the proof of \cref{lm iterate spectrum} is given in \cref{appen background proof} for completeness. 
The connection made in \cite{Knill2006,Obrien2020,seve} between phase estimation of the walk operator and expectation value estimation is manifest when considering the reflections\footnote{$\mathcal{B}[\hat{F}]$ is also a reflection due to being self-inverse and so \cref{lm iterate spectrum} can be applied.} $\mathcal{B}[\hat{F}]$ and $\textup{c}\!-R_{\sigma_0}$ with the above result. 

\begin{restatable}{thm}{thmexactestimation}\label{thm exact estimation}
Consider the walk operator $\mathcal{U}: = (\mathcal{B}[\hat{F}])\cdot(\textup{c}\mbox{-}R_{\sigma_0})$ acting on $\mathcal{H}_{\textup{sys}} \otimes \mathcal{H}_{\mathcal{B}_F}$ where $\mathcal{B}[\hat{F}]$ is a self-inverse block encoding of the renormalised observable $\hat{F}$ $(||\hat{F}||_{\textup{op}} \leq 1)$ and $\textup{c}\mbox{-}R_{\sigma_0}$ is the reflection about a groundstate, $\ket{\sigma_0}_{\textup{sys}}$, controlled on the $\ket{0}_{\mathcal{B}_F}$ subspace.
The state $\ket{\sigma_0}_{\textup{sys}}\otimes \ket{0}_{\mathcal{B}_F}$ is an equal superposition of two eigenstates of $\mathcal{U}$ with eigenphases,
\[\theta_{0,\pm} = \mp \frac{1}{2\pi} \cos^{-1}\left(\bra{\sigma_0}\hat{F}\ket{\sigma_0} \right).\]
\end{restatable}

\cref{thm exact estimation} simply applies \cref{lm iterate spectrum} to two particular reflections -- see e.g. \cite{seve} for further details of this result or \cref{appen background proof} for a proof given for self-containment.
Since $\theta_{0,\pm}$ is a simple function of the desired expectation value, $\bra{\sigma_0}\hat{F}\ket{\sigma_0}$ can be calculated from an estimate of this eigenphase.

\subsubsection*{QPE in expectation value algorithms}

In the quantum algorithm to estimate $\bra{\sigma_0}\hat{F}\ket{\sigma_0}$ there are two potential subroutines where QPE can be used. 
QPE of the walk operator is used incoherently to estimate the eigenphase $\theta_{0\pm}$, called the outer QPE or oQPE.
Additionally, QPE can be used coherently to implement the reflection about a state needed to construct the walk operator, called the inner QPE or iQPE.
The first non-asymptotic resource estimates for this routine were given in \cite{seve} where the reflection about the state is instead constructed via phase estimation due to the challenging classical pre-processing requirements of using QSP.
Both types of QPE will have finite phase registers which we will denote by Hilbert spaces $\mathcal{H}_{\text{iph}}$ ($\mathcal{H}_{\text{oph}}$) for the inner (outer) QPE.
These will both introduce bit discretisation and degrade the performance of the algorithm.

\cref{eqn QPE refl} described how QPE can be used to reflect about an eigenstate of a unitary, however $\ket{\sigma_0}$ is an eigenstate of the non-unitary Hamiltonian. 
Therefore, the iQPE will reflect about the eigenstate of the qubitised encoding of the Hamiltonian, $\mathcal{Q}[H]$, acting jointly on $\mathcal{H}_{\text{sys}}\otimes \mathcal{H}_{\mathcal{B}_H}$.
This encoding is simply related to the self-inverse block encoding (denoted $\mathcal{B}[\cdot]$) by a reflection,
\begin{equation}
    \mathcal{Q}[H] := \left(2 \ket{0}\bra{0} - \mathds{1} \right)_{\mathcal{B}_H} \otimes \mathds{1}_{\text{sys}} \cdot \mathcal{B}[H]
\end{equation}
 and therefore still encodes the operator in a subspace, $\left(\bra{0}_{\mathcal{B}_H}\otimes \mathds{1}_{\textup{sys}}\right)\mathcal{Q}[H]\left(\ket{0}_{\mathcal{B}_H}\otimes \mathds{1}_{\textup{sys}}\right) = H$.
However, instead of being self-inverse the qubitised encoding has eigenphases related to the eigenvalues of the encoded operator \cite{Low2019}.
If the Hamiltonian has eigenvectors $\{\ket{\sigma_i}\}_{i=1}^t$ with eigenvalues $\{\lambda_i\}_{i=1}^t$ then the qubitised Hamiltonian encoding has eigenvectors,
\begin{equation}
    \{\ket{\psi_{i,\pm}}: = \frac{1}{\sqrt{2}}\left(\mathds{1} \pm \ii \frac{\lambda_i \mathds{1} - \mathcal{Q}[H]}{\sqrt{1-\lambda_i}}\right)\ket{\sigma_i}_{\text{sys}}\otimes \ket{0}_{\mathcal{B}_H} \},
\end{equation}
with eigenphases $\{\phi_{i,\pm} :=\pm \frac{1}{2\pi}\cos^{-1}(\lambda_i)\}$.

\section{Algorithm analysis}\label{sect algo analysis}

Recall that ideally the expectation algorithm uses the coherent $\textup{iQPE}$ to implement a reflection about the qubitised Hamiltonian groundstate,\footnote{Here we have arbitrarily selected the positive block encoding branch. This introduces a factor of $\sqrt{1/2}$ in the success amplitude. } $\ket{\psi_{0,+}}_{\text{sys},\mathcal{B}_H}$.
However, a finite $\textup{iQPE}$ phase register only approximately achieves this, instead implementing 
\begin{equation}\label{eqn:rtildepsi}
    R_{\tilde{\psi}_{0,+}}:= \textup{QPE}^\dagger (\mathds{1} - 2\ket{\tilde{\phi}_{0,+}}\bra{\tilde{\phi}_{0,+}})_{\text{iph}}\textup{QPE}.
\end{equation}
Note that since the circuit includes a reflection about $\ket{\tilde{\phi}_{0,+}}_{\text{iph}}$ the eigenphase must be learned precisely.
Since $(R_{\tilde{\psi}_{0,\pm}})^2 = \mathds{1}$, regardless of the accuracy of the $\textup{iQPE}$, and whether techniques such as QSP or window functions are used, the implementation is still a true reflection.
Therefore, while not immediate from \cref{eqn QPE refl}, there exists a projector $\tilde{P}$ such that $R_{\tilde{\psi}_{0,+}} = \mathds{1} - 2 \tilde{P}$.
The block encoding of the observable $\hat{F}$ is self-inverse and therefore is the second reflection, with an associated projector $Q : = \frac{1}{2}\left(\mathds{1} -\mathcal{B}[\hat{F}] \right)$.
Therefore \cref{lm iterate spectrum} can be used to analyse the spectra of the resulting walk operator via the squared projector product, $\tilde{P}\cdot Q\cdot \tilde{P}$, despite the approximate implementation.

In standalone QPE -- where the goal is to estimate the eigenphase of a given eigenstate -- window functions improve the \emph{success probability} of the algorithm but leave the accuracy unchanged.
In contrast, due to the concatenated nature of using QPE as a subroutine, improving the success probability of the iQPE can lead to \emph{increased accuracy} of the oQPE. 
Finite register effects in the iQPE perturbs the reflection and affects the eigenspectra of $\tilde{P}\cdot Q\cdot \tilde{P}$ and hence $\mathcal{U}$.
Perturbation of the eigenphase of $\mathcal{U}$ translates into \emph{error} in the expectation value estimate whereas perturbation of the eigenvectors lead to decreased success probability of the oQPE.

The expectation value circuit proposed here (see \cref{fig circuit}) has two key distinctions from previous work, to leverage window-assisted QPE as a subroutine.
Firstly, the phase register is prepared with a state preparation circuit, this is more costly than simply applying Hadamards.
Secondly the observable block encoding is controlled on the phase register, introducing additional interplay between the reflections so that the preparation unitaries do not commute through the circuit and cancel.
Both these distinctions add complexity into each query to $\mathcal{U}$.
However, we will show that this is offset by the gain of implementing a $\mathcal{U}$ closer to the idealised state reflection.
Numerically this overall leads to trimmed gate counts for observable estimation. 

\subsection{Effect of window functions on the walk operator}\label{subsect window walk}

By preparing the \textup{iQPE} phase register with a unitary $\hat{W}$, window functions improve the accuracy of the reflection implementation for general QPE unitaries. 
The projector is expressed in terms of the window function,
\begin{equation}\label{eqn tilde psi messy}
    \tilde{P} = \sum_{i=0}^{2^t-1}\sum_{\{\pm\}} \ket{\psi_{i,\pm}}\bra{\psi_{i,\pm}}_{\text{sys},\mathcal{B}_H} \otimes \frac{1}{2^n} \sum_{a,a',x,x'=0}^{2^n-1}  W_{a}(x)W_{a'}^*(x') \e^{2\pi\ii ( \phi_{i,\pm} - \tilde{\phi}_{0,+})(x-x')}    \ket{a'}\bra{a}_{\text{iph}}.
\end{equation}
Let us define the useful phase register state
\begin{equation}\label{eqn rho phi}
    \ket{\rho(y)}_{\text{iph}} : = \frac{1}{\sqrt{2^n}}\sum_{a,x=0}^{2^n-1}  W_{a}^* (x) \e^{-2\pi\ii y x/2^n} \ket{a}_{\text{iph}}.
\end{equation}
Physically this is the phase register state initially in $\ket{a}_{\text{iph}}$ after the circuit for $R_\psi$ has been applied.
The value of $y$ then depends on the relative difference between the reflection phase and the phase applied by the QPE, $(\phi_{j,\pm} - \tilde{\phi}_{0,+})$, where recall that $\tilde{\phi}_{0,+}$ is the closest $n$-bit representation of $\phi_{0,+}$.
Since we are concerned with the qubitised Hamiltonian eigenbasis for the system, for convenience introduce the abbreviated notation $\ket{\rho_i^{\pm}}:=\ket{\rho(\phi_{j,\pm} - \tilde{\phi}_{0,+})}$.
The projector associated with the reflection can then be rewritten more succinctly as, 
\begin{equation}\label{eqn tilde psi clean}
    \tilde{P} = \sum_{i=0}^{2^t-1}\sum_{\{\pm\}}\ket{\psi_{i,\pm}}\bra{\psi_{i,\pm}}_{\text{sys},\mathcal{B}_H} \otimes \ket{\rho_i^\pm}\bra{\rho_i^{\pm}}_{\text{iph}}.
\end{equation}

We can now outline the expectation value algorithm -- note that before executing the algorithm, it is necessary to have a good approximation of the ground state energy, an estimate of the spectral gap, and a means of preparing the ground state.

\begin{figure}[h!]
\centering
\begin{subfigure}{.41\textwidth}
  \centering
  \includegraphics[width=0.95\linewidth]{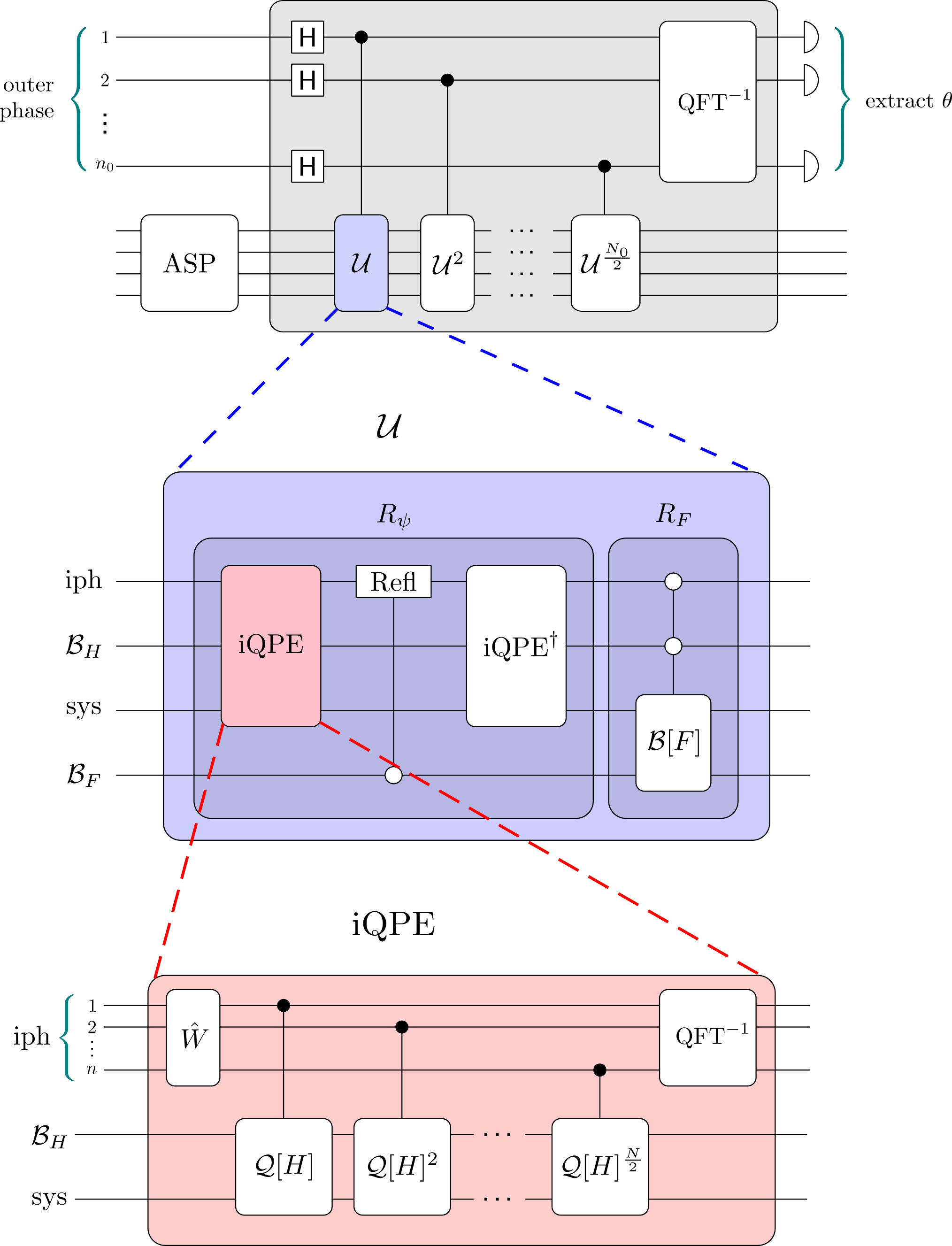}
  \newsubcap{Circuit for \cref{algo label}. }
  \label{fig circuit}
\end{subfigure}%
\begin{subfigure}{.5\textwidth}
  \centering
  \begin{mdframed}
\begin{restatable}{alg}{alg1}\label{algo label}
Expectation value estimation: estimate $\bra{\psi_0}\hat{F}\ket{\psi_
0}$ to within $\epsilon$ with success probability $p_\textup{success}$
\end{restatable}
\begin{enumerate}
    \item Prepare the state $\ket{\psi_{0,+}}_{\text{sys},\mathcal{B}_H} \otimes \ket{0}_{\mathcal{B}_F}$ using any arbitrary state preparation routine (ASP).
    \item Run the circuit, replacing $R_F$ by $\mathds{1} - 2 \ket{00}\bra{00}_{\mathcal{B}_H,\mathcal{B}_F}$ in $\mathcal{U}$ to estimate $\theta'_{0,\pm} = \pm \frac{1}{\pi}\cos^{-1}(2|\braket{\rho_0^+}{0}|^2 -1)$
    \item Classical processing: from $\theta'_{0,\pm}$ calculate the overlap $\abs{\braket{\rho_0^+}{0}}$
    \item Run the circuit using the observable $\hat{F}$ to estimate $\theta_{0,\pm}$ from \cref{lm U spectra}
    \item Classical processing: from $\theta_{0,\pm}$ calculate $\bra{\sigma_0}\hat{F}\ket{\sigma_0}$ up to error $\epsilon$ (\cref{prop error})
\end{enumerate}
\end{mdframed}
\end{subfigure}
\end{figure}

An understanding of the eigenspectra of $\mathcal{U}$, when window-assisted \text{iQPE} is used to perform the reflection about the state, is essential to quantify the error in the algorithm.
The following result describes the eigenphases of $\mathcal{U}$.

\begin{restatable}{lemma}{lemmaUspec}\textup{(Eigenspectra of walk operator)}\label{lm U spectra}
Let $\hat{F}$ be an observable  ($||\hat{F}||_{\textup{op}}\leq 1$) and $H$ be a  Hamiltonian ($||H||_{\textup{op}}\leq 1$).
Consider the unitary $\mathcal{U} := (\textup{c}\mbox{-}\mathcal{B}[\hat{F}])\cdot(\textup{c}\mbox{-}R_{\tilde{\psi}_{0,+}})$, a rotation constructed from two controlled reflections.

Then $\mathcal{U}$ has eigenphases $\theta_{0,\pm} = \pm \frac{1}{2\pi} \cos^{-1}\left(2\abs{\omega}^2 -1 \right)$ where
\[\frac{1}{4}\bra{\sigma_{0}}\mathds{1} - \hat{F} \ket{\sigma_0}_{\textup{sys}}\, |\braket{\rho_0^+}{0}_{\textup{iph}}|^2 - \nu \leq \abs{\omega}^2  \leq \frac{1}{4}\bra{\sigma_{0}}\mathds{1} - \hat{F} \ket{\sigma_0}_{\textup{sys}}\, |\braket{\rho_0^+}{0}_{\textup{iph}}|^2 + \nu;\]
\[\nu:= 2 |\braket{\rho_0^+}{0}|\max_{j>0,\pm} \left\{\abs{\braket{0}{\rho_j^\pm}_{\textup{iph}}},\abs{\braket{0}{\rho_0^-}_{\textup{iph}}} \right\},\]
with $\ket{\rho_{j}^\pm}_{\textup{iph}} = \frac{1}{\sqrt{2^n}}\sum_{a,x=0}^{2^n-1} W_{a}^* (x) \e^{-2\pi\ii (\phi_{j,\pm} - \tilde{\phi}_{0,+}) x}\ket{a}_{\textup{iph}}$.
$\textup{c}\mbox{-}\mathcal{B}[\hat{F}]$ denotes the block encoding $\mathcal{B}[\hat{F}]$ controlled on the $\ket{00}_{\textup{iph},\mathcal{B}_H}$ subspace. 
$\ket{\sigma_0}_{\textup{sys}}$ is the groundstate of $H$ with the eigenphase $\phi_{0,+}$ corresponding to the positive branch of the qubitised encoding of the Hamiltonian $\mathcal{Q}[H]$.
The reflection $R_{\tilde{\psi}_{0,+}}$ (\cref{eqn:rtildepsi}) is controlled on the $\ket{0}_{\mathcal{B}_F}$ subspace and implemented using a coherent $\textup{QPE}$ with $n$ phase qubits prepared by $\hat{W}$.
\end{restatable}

The proof of \cref{lm U spectra} is given in \cref{appen background proof}.
With an exact reflection the above resembles \cref{thm exact estimation}, $\abs{\omega}^2 = \frac{|\braket{\rho_0^+}{0}_{\textup{iph}}|^2}{4}\bra{\sigma_{0}}\mathds{1} - \hat{F} \ket{\sigma_0}$\footnote{The additional factor of $|\braket{\rho_0^+}{0}_{\textup{iph}}|^2$ is a result of the extra control added to the observable block encoding. In the case of an infinite QPE phase register $|\braket{\rho_0^+}{0}_{\textup{iph}}|^2\rightarrow 1$. }.
Without making any assumptions about the distribution of the Hamiltonian spectra, \cref{lm U spectra} describes the propagated effect of the imperfect iQPE on the iterate.
The error depends on the overlap $|\braket{0}{\rho_j^\pm}|$ which is affected by: the number of qubits $n$ in the QPE, the state preparation $\hat{W}$ and the phase difference $(\phi_{j,\pm} - \tilde{\phi}_{0,+})$. 

For a useful algorithm, the Hamiltonian is assumed to have a known lower bound on the spectra gap $\delta$ so that the ground state can be distinguished.
The Hamiltonian is then sub-normalised during the block encoding so that the maximum eigenvalue of the block encoded operator is $(1-\delta_r)$ where $\delta_r = \delta/\norm{H}_{\text{op}}$ is the relative spectral gap.
The assumption also implies a known lower bound on the phase gap $\delta_\phi = \abs{\phi_{1,+} - \phi_{0,+}}$ for the qubitised Hamiltonian.

The size of the inner phase register is determined by the phase gap.
Let $l$ be the minimum number of qubits required to resolve the phase gap, and add $m$ extra qubits to the inner phase register to improve the performance so that the total is $n = l+m$.
To successfully target the reflection, the reflected phase must be the closest finite bit representation of the true eigenphase of $\ket{\psi_{0,+}}$ i.e. $|\phi_{0,+} -\tilde{\phi}_{0,+}|< \frac{1}{2^{l+m}}$ and all other eigenphases are well separated from the reflected phase $|\phi_{1,+} - \tilde{\phi}_{0,+}| > \frac{1}{2^{l+m}}$ for all $m$.
These two conditions give us the following relationship between the spectral gap promise and the number of qubits $\delta_{\phi} > \frac{1 + 2^{m}}{2^{l+m}}$.
This ensures that $|\braket{0}{\rho_j^\sigma}|$ is close to 1 only when $j=0$ and $\sigma=+$ and close to zero for \emph{all excited states}.
The overlap can be suppressed for states well separated from $\tilde{\phi}_{0,+}$ by increasing $m$, illustrated by \cref{fg overlap} for the rectangular window. 
This figure also illustrates that the value of the overlap for contaminating states depends less on the absolute difference in phase $|\phi - \tilde{\phi}_{0,+}|$ more on whether the difference is close to a $(l+m)$-bit representation. 
Therefore, only considering the low energy excited states does not give an accurate estimate of the error and we analyse the full Hamiltonian spectrum to accurately describe the error.

\begin{figure}[h!]
\centering
\includegraphics[trim={0cm 0cm 0cm 0cm},width=0.55\textwidth]{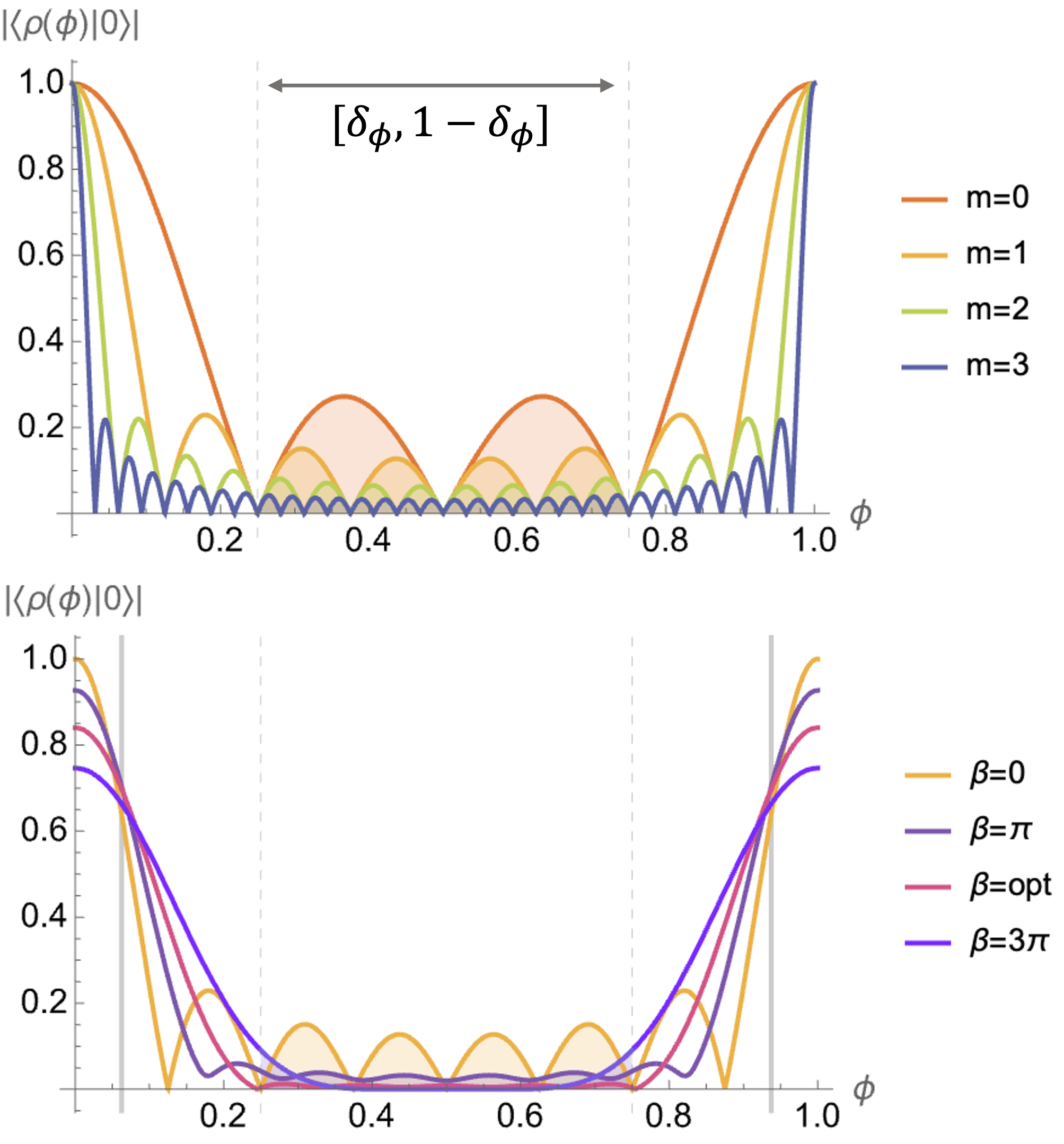}
\caption{Plot of the behaviour of the overlap, $|\braket{\rho(\phi)}{0}|$, for different phase register size and preparation. Recall that $\ket{\rho(\phi)}$ is defined in \cref{eqn rho phi} and corresponds to the state of the phase register after performing a QPE and a reflection on the phase register. The maximum absolute value of the overlap appears in the error bound in \cref{lm U spectra}. The spectral gap of the Hamiltonian ensures that for the error bound only the overlap in the range $[\delta_\phi, 1-\delta_\phi]$ needs to be considered. The dashed grey lines demark this cut off where $\delta_\phi = 1/2^l$ where $l$ is the minimum number of phase qubits, and the region contributing to the error is shaded under the curve. The smaller the maximum value of the overlap in the shaded region the more accurate the QPE implementation of the reflection. \textbf{Top}: Adding additional phase register qubits reduces the overlap given a rectangular window. The total number of phase qubits is $n = l+m$, and the phase register is prepared in uniform superposition. \textbf{Bottom}: Preparing the phase register of the QPE in a Kaiser window reduces the overlap. The number of phase qubits is fixed with $m=1$. Increasing $\beta$ widens the central lobe and suppresses the side-bands of the overlap function. There is a close to optimal $\beta$ taken from \cref{lm for opt beta} that reduces the overlap maximum above the dashed line. The thick grey line demarks the maximum value of $|\phi_0^+ - \tilde{\phi}_0^+|$.}
\label{fg overlap}
\end{figure}

The overlap value also depends on the choice of phase preparation unitary $\hat{W}$.
Since aside from the spectral gap we assume no knowledge of the Hamiltonian spectrum by choosing a tapered function -- such as the Kaiser window (\cref{eqn Kaiser}) --  this overlap can be suppressed for all excited states.
This is done by increasing the tunable bandwidth parameter $\beta$ to suppress the DFT side-band amplitude at the expense of widening the central lobe -- refer back to \cref{fg FT of QPE}.
Complexity analysis \cite{Berry} shows that in the asymptotic limit increasing $\beta$ exponentially suppresses the QPE failure probability.\footnote{The error in \cref{lm U spectra} is connected to the maximum value of the overlap outside of the spectral gap region. The plot of the overlap squared corresponds to the success probability in simple QPE, therefore the failure probability that is exponentially suppressed here corresponds to the areas under the curve outside of the spectral gap region -- a related but different quantity to that which is studied here.}
In practice, when there is a finite phase register there is a trade-off to be found when choosing $\beta$: increasing $\beta$ suppresses the side bands so initially decreases the value of the overlap, however if $\beta$ is increased further the central lobe widens enough to extend past the spectral gap region and increase the value of the overlap. 
One method for choosing $\beta$ for a given $l$ and $m$ is by matching the width of the central lobe to the spectral gap region. 
The width of the central lobe is given by the first minima of the DFT:

\begin{restatable}{claim}{claimbeta}\label{lm for opt beta}
Consider that a window-enhanced QPE is used to reflect about the ground state of a qubitised Hamiltonian with spectral phase gap lower bounded by $2^{-l}.$
The $(l+m)$-qubit phase register of a QPE is prepared in the superposition $\sum_{x=-(N-1)/2}^{(N-1)/2}K(x,\beta)\ket{x}$ where $K(x,
\beta) \propto I_0(\beta \sqrt{1 - (x/2N)^2})$ is the Kaiser function with bandwidth $\beta$ and $N=2^{l+m}$.
The value of $\beta$ that minimises the contamination error from excited states is
\begin{equation}
    \beta_{\text{opt}} = \pi \sqrt{2^{2m}-1}.
\end{equation}
\end{restatable}

\begin{proof}
The Fourier transform of the Kaiser function is proportional to $\sinc[\sqrt{(N\omega/2)^2 - \beta^2}]$.
Therefore, the first minima of the Fourier transform occurs at $\sqrt{(N\omega_\text{min}/2)^2 - \beta^2} = \pi$.
Solving for $\omega_{\text{min}}$, the phase of the first minima in the DFT $\phi_k$ is given by
\begin{equation*}
    2\pi \phi_K = \sqrt{4(\pi^2 + \beta^2)}/N.
\end{equation*}
Matching $\phi_K = \frac{1}{2^l}$ gives an analytic expression for the `optimal' $\beta$.
\end{proof}

For $m>0$ a non-zero value for $\beta$ outperforms the rectangular function\footnote{Generally at least one extra qubit is required to determine the phase \cite{nielsen2010quantum}.} as seen in \cref{fg overlap}.
This optimal bandwidth scales as $\beta \in O(2^{m})$ as expected from the complexity analysis in \cite{Berry}.
Optimality is quoted since in practice the best $\beta$ to minimise the resource count will depend on several factors not only minimising overlap, for example increasing $\beta$ also has the effect of reducing the peak amplitude at null phase difference and may be harder to calculate and prepare the state.
To optimise resource counts for real applications it could be beneficial to set a maximum value of $\beta$ and follow \cref{lm for opt beta} only up to this threshold due to the expense of preparing a high $\beta$ window function on many qubits, this is done to obtain the resource estimates in later sections. 

\subsection{Error and success probability}

This section analyses the error and success probability of \cref{algo label}, starting with the error analysis. 
There are several sources of error inherent to the algorithm arising from finite register size, even assuming perfect gate implementations and exact state preparation routines. 
With these assumptions, algorithmically the error in the final expectation value estimate is a combination of contributions: \emph{reflection error}, from the imperfect implementation of the state reflection due to the finite iQPE phase register; \emph{oQPE error}, from the imprecision in the output due to the oQPE and \emph{pre-learning error} from learning the value of $|\braket{\rho_0^+}{0}|^2$.
These error contributions amalgamate in the following algorithmic error bound.

\begin{restatable}{prop}{properror}\textup{(Expectation value algorithm error analysis)}\label{prop error}
A successful instance of \cref{algo label} outputs an $\epsilon$-approximation of $\bra{\sigma_0}\hat{F}\ket{\sigma_0}$ where
\[\epsilon \leq \frac{3\pi 2^{-n_o}+ 8|\braket{\rho_0^+}{0}|\max_{j>0,\pm} \left\{|\braket{0}{\rho_j^\pm}|,|\braket{0}{\rho_0^-}| \right\} + c }{|\braket{\rho_0^+}{0}|^2},\]
given $\ket{\rho_{j}^\pm}_{\textup{iph}} := \frac{1}{\sqrt{2^n}}\sum_{a,x=0}^{2^n-1} W_{a}^* (x) \e^{-2\pi\ii (\phi_{j,\pm} - \tilde{\phi}_{0,+}) x}\ket{a}_{\textup{iph}}$, the number of qubits in the phase register of the \textup{oQPE} is $n_o$ and $c$ is a free non-negative constant.
\end{restatable}

Proof of \cref{prop error} deferred to \cref{appen background proof}

The free constant $c$ appears in the bounds for both the error and the success probability -- hence there is a tradeoff to be considered when setting its value.  
This error bound is needed to estimate the computational resources of the algorithm in \cref{sect RE}.
The success probability of the algorithm can also be bounded:

\begin{restatable}{prop}{propsuccess}\textup{(Expectation value algorithm success probability)}\label{prop success}
\cref{algo label} succeeds with probability,
\[p_{\textup{success}}\geq |\braket{0}{\rho_0^+}|^2 \left(1 - \frac{32 \max_{j>0,\pm} \left\{|\braket{0}{\rho_j^\pm}|,|\braket{0}{\rho_0^-}| \right\}^2}{c}\right),\]
where $\ket{\rho_{j}^\pm}_{\textup{iph}} := \frac{1}{\sqrt{2^n}}\sum_{a,x=0}^{2^n-1} W_{a}^* (x) \e^{-2\pi\ii (\phi_{j,\pm} - \tilde{\phi}_{0,+}) x}\ket{a}_{\textup{iph}}$ and $c$ is the free non-negative constant from the error tolerance.
\end{restatable}

Referring to \cref{fg overlap} if the system is gapped and an appropriate number of QPE qubits are used this ensures $\max_{(j,\sigma) \neq (0.+)} \left\{|\braket{0}{\rho_j^\pm}|,|\braket{0}{\rho_0^-}| \right\}$ is small and the success probability is close to 1.
See \cref{appen background proof} for the full proof of \cref{prop success}. 
The free constant $c$ in the error and success probability can be freely tuned, but for example if $c = |\braket{\rho_0^+}{0}|\max_{j>0,\pm} \left\{|\braket{0}{\rho_j^\pm}|,|\braket{0}{\rho_0^-}| \right\}$ then 
\begin{align}
    \epsilon &\leq \frac{3\pi 2^{-n_o}+ 9|\braket{\rho_0^+}{0}|\max_{j>0,\pm} \left\{|\braket{0}{\rho_j^\pm}|,|\braket{0}{\rho_0^-}| \right\}}{|\braket{\rho_0^+}{0}|^2}\\
    p_{\textup{success}}&\geq |\braket{0}{\rho_0^+}|^2 \left(1 - \frac{32 \max_{j>0,\pm} \left\{|\braket{0}{\rho_j^\pm}|,|\braket{0}{\rho_0^-}| \right\}}{|\braket{0}{\rho_0^+}|}\right).
\end{align}

\section{Algorithm resource estimates}\label{sect RE}

\cref{prop error} shows that the final error in the expectation value is comprised of an error due to the oQPE (that is reduced by increasing the number of oQPE phase qubits) and an error due to the reflection about the eigenstate (that is reduced by either increasing the number of iQPE phase qubits or using QSP).
When using the Kaiser window there is a reduction in gate count that originates from reducing this reflection error. 
\cref{fig RE} illustrates that for the same number of quantum resources, a Kaiser window function with optimal bandwidth achieves a lower error in the reflection implementation than a rectangular window and QSP-rounded QPE for a large range of relative errors. 
While the scaling of the QPE with window functions is still linear (whereas with QSP the error decreases exponentially with the increased resources) the constants involved lead to window functions outperforming QSP for most practical cases.

\begin{figure}[h!]
\centering
\includegraphics[trim={0cm 0cm 0cm 0cm},width=0.7\textwidth]{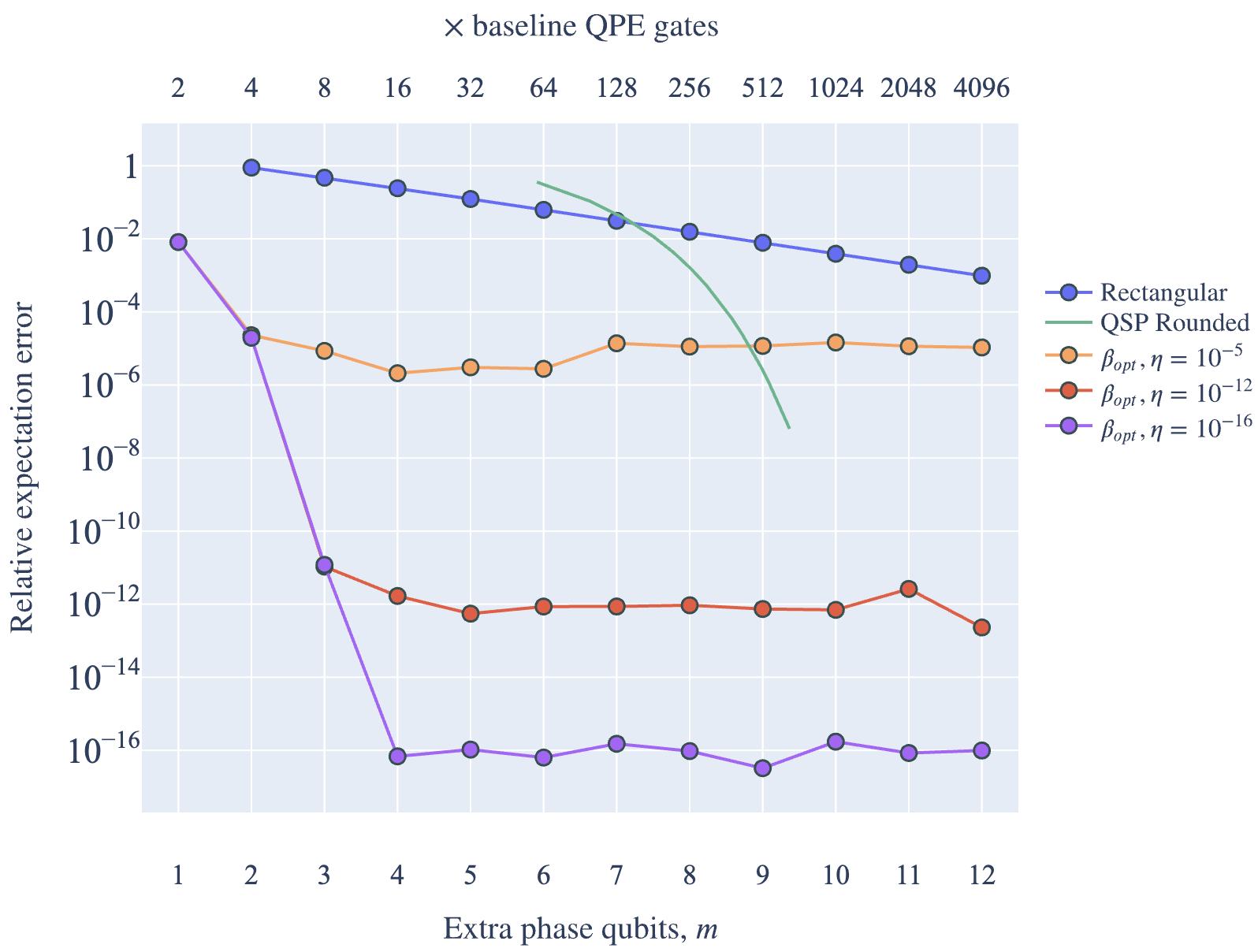}
    \caption{Plot of the relative expectation error arising from the reflection implementation. Standard QPE (`Rectangular') and using QSP to round the QPE (`QSP rounded') is compared to Kaiser window-enhanced QPE proposed in this work with optimal bandwidth. Three Kaiser window cases are plotted with varying error in the state preparation, $\eta$. In the standard and Kaiser window-enhanced versions, the error is reduced by adding qubits to the phase register. The size of the phase register determines the number of calls to the controlled unitary (in this case the block encoded Hamiltonian), hence this reduction in error leads to increased gate count reflected by the top $x$ axis. In the QSP case the error is reduced by performing a higher degree polynomial, this also corresponds to increasing the quantum resources but does not effect the size of the phase register, hence only the top $x$ axis corresponds to the QSP curve (opposed to the other traces where both $x$-axis apply). The QSP curve represents the minimum error achievable, but in practice finding the phase angles to the required accuracy may be challenging. One baseline QPE query corresponds to the number of quantum gates required to perform QPE on the system with the minimum number phase qubits sufficient to discern the ground state. Increasing the number of phase qubits linearly decreases the relative expectation error on the log plot, in the Kaiser case the linear trend is disrupted when the relative error is comparable with the state preparation error. While the QSP rounded version decreases exponentially according to the emulation, in practice requires heavy classical pre-processing to calculate phase angles. Nevertheless due to much steeper gradient of the linear Kaiser window curve, the window-enhanced version also outperforms QSP for most practical cases. }
    \label{fig RE}
\end{figure}

As a case study for the use of window functions in coherent QPE, we demonstrate that constant factor resource estimates for the expectation value estimation algorithm can be significantly reduced by using the Kaiser window function. The expectation value of both one and two body observables are considered with respect to the ground state of different molecular systems. For the benchmark dataset we selected three closed-shell molecules — water, ammonia, and p-benzyne — treating each in the all-electron picture with a cc-pVDZ basis set. Additionally, we consider a cytochrome P450 heme model in an open-shell configuration, using a (42e, 43o) active space in the cc-pVDZ basis targeting the $S = 1$ spin state. In previous work, this active space model was chosen to elucidate the  mechanism by which artemisinin binds to heme \cite{cortes2024fault, posner2004knowledge}.

To produce a resource estimate of the expectation value algorithm, the Hamiltonian of the given molecular system must be block-encoded.
In previous work, double factorisation was used to block-encode both the electronic structure Hamiltonian and relevant observables \cite{seve}, however, further cost reductions may be achieved with symmetry-shifting or block-invariant symmetry-shifting (BLISS) techniques \cite{Rocca2024,loaiza2023block, Zak2024}.
This work advocates for the use of tensor hypercontraction (THC) which has been shown to provide improved resource reduction in the estimation of molecular energies \cite{Evenmoreefficient, blissthc}, and finds that BLISS-THC provides a significant improvement in the block-encoding cost for observable estimation compared to previous approaches.
The tensor factorization procedure was adapted to the expectation value algorithm by concurrently optimizing the rank of the Hamiltonian and observable  with respect to a precision budget which estimates the error of the expectation value obtained by the truncated Hamiltonian and truncated observable against that of their respective untruncated versions using coupled cluster calculations with singles and doubles (CCSD) \cite{sun2020recent}.
Details of the procedure may be found in \cref{appen:error handling}.
This strategy leads to a further reduction in the overall costs in comparison to previous numeric end-to-end resource estimates for this routine \cite{seve}, evidenced by the rectangular window and QSP-rounded traces. 

\begin{figure}[h!]
\centering
\includegraphics[trim={0cm 0cm 0cm 0cm},width=0.9\textwidth]{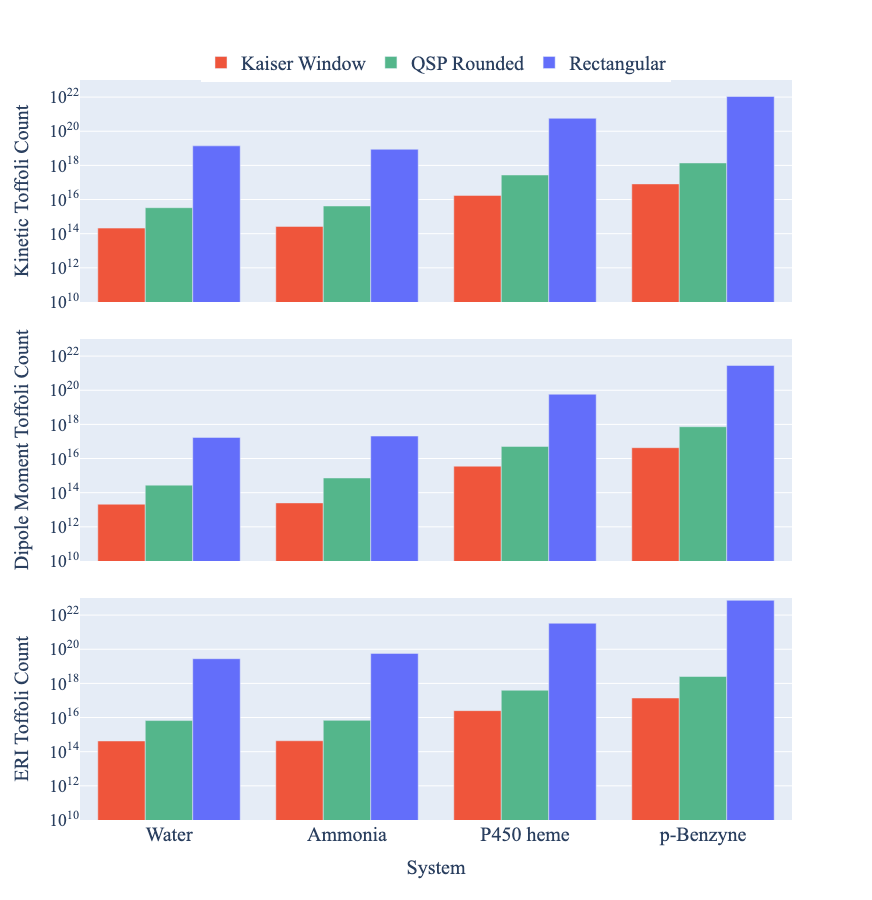}
    \caption{Resource estimates for the expectation value algorithm for a variety of systems: water, ammonia, p-Benzyne, P450 heme; and a variety of observables: kinetic energy, dipole moment and electron repulsion integral (ERI). The left axis denotes the estimated number of Toffolis gates required for the expectation value algorithm. Three methods for implementing the inner QPE also depicted in \cref{fig RE} are given to demonstrate the impact on the overall algorithm of having a higher success amplitude in the coherent QPE. For all Hamiltonian block-encodings BLISS-THC is used, as well as for ERI the two-body observable.}
    \label{fig RE numeric}
\end{figure}

The complete quantum resource estimation for this algorithm in terms of Toffoli gate counts is plotted in \cref{fig RE numeric} for various molecules and observables including the kinetic energy, electric dipole moment, and electron repulsion interaction operators.
The target error tolerances used are $\epsilon_{\text{kin.}} = 1.6 \text{ mHa}$, $\epsilon_{\text{dip.}} = 10\text{ mDebye}$ and $\epsilon_{\text{eri}} = 1.6 \text{ mHa}$ consistent with common literature values.
\cref{appen RE} gives more details into how these resource estimates are calculated.
For all algorithm variants the gate count increases with the complexity of the molecule as expected, but a significant reduction is seen using window-assisted iQPE compared to QSP rounding -- in the largest system an almost 2-orders-of-magnitude reduction in the Toffoli count is observed.
Compared to the naive implementation of QPE, the proposed routine represents a reduction in the Toffoli gate count by more than five orders of magnitude. 

The logical qubit cost of the algorithm scales linearly with the Hamiltonian size -- estimates given in \cref{table logical qubits}.
The qubit count is largely comparable across the different variations since the qubit cost is dominated by the block encoding ancilla requirements.
The rectangular window case requires systematically slightly more qubits than the other variants, and the QSP estimate is slightly marginally lower ($~5$ less) than the window case.
Future work will examine how data batching and additional space–time trade-offs in the block-encoding of the Hamiltonian and observables can further reduce the logical-qubit count.

\begin{table}[h!]
    \centering
    \begin{tabular}{c*4r*2r*2r}
\hline \hline 
\bf{System} & \bf{Observable} & \multicolumn{3}{c}{\bf{Logical qubits}}  \\
 & & Rectangular & QSP rounded & Kaiser window  \\
\hline
 & KE & 520 & 499 & 504 \\
$\text{Water}$ & $x$-dipole & 429 & 411  & 416 \\
& ERI & 518 & 497  & 502 \\
\hline
 & KE & 628 & 608 & 613 \\
$\text{Ammonia}$ & $x$-dipole & 494 & 477  & 481\\
&  ERI & 586 & 564  & 569 \\
\hline
 & KE & 1096 & 1076 & 1081 \\
$\text{P450 heme}$& $x$-dipole & 910 & 891  & 896 \\
& ERI & 819 &  797 & 802 \\
\hline
 & KE & 2646 & 2624 & 2629 \\
$\text{p-Benzyne}$& $x$-dipole & 2708 & 2687  & 2692 \\
& ERI & 2582 & 2558  & 2563 \\
\hline
 \hline 
\end{tabular}
    \caption{Table of logical qubit highwater for the algorithm. Various subroutines require local ancilla that are uncomputed at the end of the subroutine so can be reused for later routines, the logical qubit highwater corresponds to the maximum number of logical qubits required in parallel at any point in the algorithm. Details of this estimation are given in \cref{appen RE}.  }\label{table logical qubits}
\end{table}

\section{Conclusions}
This work builds on \cite{Knill2006,Obrien2020,seve} to present an adapted expectation value estimation algorithm that employs window-assisted coherent QPE.
This modification results in improved algorithmic-error-to-quantum-resources ratio that can be seen in \cref{fg overlap}.
Since QPE is one of the small handful of elemental quantum routines, we hope this analysis will act as a model for how window-assisted coherent QPE can be used as a subroutine and lead to reduced overall errors, for instances in other cases where reflection about an eigenstate is used.
Additionally the classical preprocessing of the block encodings of the Hamiltonian and observable together for expectation value estimation is novel to this work and highlights the potential for further optimisation of error budgeting.  
While estimating molecular observables is still an expensive algorithm, the ability to produce end-to-end resource estimates and see these estimates reduced is very encouraging and a further step on the path towards closing the gap between commercially interesting applications and hardware capabilities.

\section*{Acknowledgments}
The authors would like to thank the team at PsiQuantum for various discussions and support.
The work was facilitated by access to PsiQuantum software suit including the visualisation tools which produced \cref{fig circuit}.

\bibliographystyle{plain}
\bibliography{references}

\newpage
\begin{appendix}
\noindent\LARGE \textbf{Appendices}\normalsize
\section{Deferred proofs}\label{appen background proof}
To aid the readability of the main text, the proofs of the claims are collated here, including some intermediate technical results that are used in the proofs.

\subsection{Proof of \cref{lm iterate spectrum}}

The rest of the proofs make good use of the following general result about the product of two reflections:
\lmiteratespec*
\begin{proof}
Jordan's lemma \cite[Chapter VII]{bhatia1996matrix} tells us that given any two projectors $\Pi_A$ and $\Pi_B$ the Hilbert space can be block diagonalised into two-dimensional subspaces $S_k$ invariant under the action of both projectors and one-dimensional subspaces $T_j$ on which $\Pi_A\Pi_B$ acts as either the identity of the null projector. 
The projectors can then be expressed in the basis that respects this block diagonal structure:
\[    \Pi_A = \sum_k \ketbra{a_k}{a_k}\otimes\ketbra{k}{k}  + \pi_B^\perp    \qquad \text{and} \qquad
    \Pi_B  =\sum_k \ketbra{b_k}{b_k}\otimes\ketbra{k}{k} + \pi_A^\perp,\]
where $\braket{b_k}{a_k} = \omega_k\in \mathbb{R}$  and $\braket{k}{j}=\delta_{j,k}$.
Additionally $\pi_A^\perp$ ($\pi_B^\perp$) is orthogonal to $\sum_k \ketbra{k}{k}$ and $\Pi_A$ ($\Pi_B$).
Each subspace $S_k$ is spanned both by $\{\ket{a_k}, \ket{a_k^\perp} \}$ and $\{\ket{b_k}, \ket{b_k^\perp} \}$ where $\ket{a_k^\perp} = \frac{1}{\sqrt{1-\omega_k^2}}\left(\mathds{1} - \ket{a_k}\bra{a_k}\right)\ket{b_k}$ and vice verse. 
So we have the relation $\ket{b_k} = \omega_k \ket{a_k} + \bar{\omega}_k\ket{a_k^\perp}$ where $\bar{\omega}_k = \sqrt{1-\abs{\omega_k}^2}$.
With this structure in place we can analyse the action of $\mathcal{U} = R_BR_A$:
\begin{subequations}
\begin{align}
    \mathcal{U} &= \left(\I - 2\left(\sum_k \ketbra{b_k}{b_k}\otimes \ketbra{k}{k} \right)+ \pi_A^\perp\right)\cdot \left(\I - 2\left(\sum_k \ketbra{a_k}{a_k}\otimes \ketbra{k}{k} \right)+ \pi_B^\perp\right)\\
    & = \I + \sum_k \left(4 \omega_k \ketbra{b_k}{a_k} - 2 \ketbra{b_k}{b_k} - 2 \ketbra{a_k}{a_k} \right)\otimes \ketbra{k}{k} - 2\pi_A^\perp - 2 \pi_B^\perp\\
    & = \I  + \sum_k\left(4\omega_k(\omega_k \ket{a_k} + \bar{\omega}_k\ket{a_k^\perp})\bra{a_k} - 2 (\omega_k\ket{a_k} + \bar{\omega}_k \ket{a_k^\perp})(\omega_k \bra{a_k} + \bar{\omega}_k \bra{a_k^\perp}) - 2 \ketbra{a_k}{a_k} \right)\otimes \ketbra{k}{k}\notag\\
    &\hspace{10mm}- 2\pi_A^\perp - 2\pi_B^\perp
\end{align}
\end{subequations}

Then we show that the following states are eigenstates of $\mathcal{U}$:
\begin{equation}
    \ket{v_{k,\pm}} = \frac{1}{\sqrt{2}}\left(\ket{a_k} \pm \ii \ket{a_k^\perp} \right)\otimes \ket{k}
\end{equation}
with eigenphases given by the stated result. 
\begin{subequations}
\begin{align}
    \mathcal{U}\ket{v_{k,\pm}} &= \mathcal{U} \frac{1}{\sqrt{2}}\left(\ket{a_k} \pm \ii \ket{a_k^\perp} \right)\otimes \ket{k}\\
    & = \frac{1}{\sqrt{2}}\left( 2 \abs{\omega_k}^2-1\pm 2\ii \omega_k \bar{\omega}_k\right)\ket{a_k}\otimes\ket{k} \pm  \frac{\ii}{\sqrt{2}}\left( 2\abs{\omega}_k -1 \mp 2 \ii \omega_k \bar{\omega}_k\right)\ket{a_k^\perp}\otimes\ket{k}\\
    & = \left(2\abs{\omega_k}^2 -1 \mp 2\ii \omega_k \bar{\omega}_k \right)\ket{v_{k,\pm}}\\
    & = \left(\cos(2\pi \theta_{k,\pm})+ \ii \sin(2\pi \theta_{k,\pm}) \right)\ket{v_{k,\pm}}.
\end{align}
\end{subequations}
Thus demonstrating that $\ket{v_{k,\pm}}$ is indeed an eigenvector of $\mathcal{U}$ with eigenphase $\theta_{k,\pm}$.
\end{proof}

\subsection{Proof of \cref{thm exact estimation}}
\thmexactestimation*
\begin{proof}
Let $\textup{c}\mbox{-}R_{\sigma_0} = \mathds{1} - 2 \ket{\sigma_0}\bra{\sigma_0}_{\text{sys}}\otimes \ket{0}\bra{0}_{\mathcal{B}_F}$ be the first reflection in \cref{lm iterate spectrum} and $\mathcal{B}[\hat{F}] = \mathds{1} - 2 Q$ be the second. 
Hence we have the associated projectors $P: = \ket{a_0}\bra{a_0} = \ket{\sigma_0}\bra{\sigma_0}_{\text{sys}}\otimes \ket{0}\bra{0}_{\mathcal{B}_F}$ and $Q: = \frac{1}{2}(\mathds{1} - \mathcal{B}[\hat{F}])$.
The singular value decomposition of the product of projectors is 
\begin{equation}
    P\cdot Q=\left(\ket{\sigma_0}\bra{\sigma_0}_{\text{sys}}\otimes \ket{0}\bra{0}_{\mathcal{B}_F}\right)\cdot\frac{1}{2}(\mathds{1} - \mathcal{B}[\hat{F}]) = \sum_k \omega_k \ket{a_k}\bra{b_k}.
\end{equation}
The singular values can be found by considering the eigenvalues of the modulus square:
\begin{subequations}
\begin{align}
    P \cdot Q \cdot Q^\dagger \cdot P^\dagger & = P \cdot Q \cdot P\\
    & = \bra{\sigma_0}\bra{0} \frac{1}{2}(\mathds{1} - \mathcal{B}[\hat{F}]) \ket{\sigma_0}\ket{0} \quad  \ket{\sigma_0}\bra{\sigma_0}_{\text{sys}}\otimes \ket{0}\bra{0}_{\mathcal{B}_F}\\
    & = \frac{1}{2}\left(1 -\bra{\sigma_0}\hat{F}\ket{\sigma_0}\right) \quad \ket{\sigma_0}\bra{\sigma_0}_{\text{sys}}\otimes \ket{0}\bra{0}_{\mathcal{B}_F}\\
     &= \abs{\omega_0}^2 \ket{a_0}\bra{a_0}.
\end{align}
\end{subequations}
This is a simple case since $\hat{a}$ is a rank one projector with $\ket{a_0} = \ket{\sigma_0}_{\text{sys}}\ket{0}_{\mathcal{B}_F}$ and so there is only one non-trivial squared singular value $\abs{\omega_0}^2 = \frac{1}{2}\left(1 -\bra{\sigma_0}\hat{F}\ket{\sigma_0}\right)$.
Along with \cref{lm iterate spectrum} it follows that the only two eigenstates of $\mathcal{U}$ with non-trivial eigenvalues are:
\begin{equation}
    \ket{v_{\pm}} = \frac{1}{\sqrt{2}} \left(\ket{\sigma_0}_{\text{sys}}\ket{0}_{\mathcal{B}_F} \pm \ii \ket{\perp}\right)
\end{equation}
with eigenphases $\theta_{\pm} = \mp \frac{\ii}{2\pi} \cos^{-1}\left(\bra{\psi_0}\hat{F}\ket{\psi_0} \right)$.
Then the statement is validated by noting the state's eigenvector decomposition  $\ket{\sigma_0}_{\text{sys}}\ket{0}_{\mathcal{B}_F} = \frac{1}{\sqrt{2}}\left(\ket{v_+} + \ket{v_-} \right)$.
\end{proof}

\subsection{Proof of \cref{lm U spectra}}

Note that the unitary anaylsed below $\mathcal{U}$ is depicted in the purple box of \cref{fig circuit}.

\lemmaUspec*

This proof of the above utilises the following intermediate results.

\begin{claim}\label{claim orthog}
$\{\ket{\psi_{i,\pm}}\ket{\rho(\phi_{i,\pm}-\frac{m}{2^n})}\}_{i,\pm,m}$ is a complete orthogonal basis for $\mathcal{H}_{\textup{sys}}\otimes \mathcal{H}_{\textup{iph}}\otimes \mathcal{H}_{\mathcal{B}_H}$, where
\[\{\ket{\psi_{i,\pm}}: = \frac{1}{\sqrt{2}}\left(\mathds{1} \pm \ii \frac{\lambda_i \mathds{1} - \mathcal{Q}[H]}{\sqrt{1-\lambda_i}}\right)\ket{\sigma_i}_{\textup{sys}}\otimes \ket{0}_{\mathcal{B}_H} \}\]
and 
\[\ket{\rho(y)}_{\textup{iph}} : = \sum_{a=0}^{2^n-1}f(y,a)\ket{a}_{\textup{iph}},\]
where $f(y, a): = \frac{1}{\sqrt{2^n}}\sum_{x=0}^{2^n-1} W_{a}^* (x) \e^{-2\pi\ii y x}$ and $\{\sigma_i \}$ are the eigenvectors of an operator $H$ acting in $\mathcal{H}_{\textup{sys}}$ with eigenvalues $\lambda_i$ and $\mathcal{Q}[H]$ is a qubitised encoding of $H$.
\end{claim}
\begin{proof}
$\braket{\psi_{i,\pm}}{\psi_{j,\pm'}} = \delta_{i,j}\delta_{\pm,\pm'}$ by definition since $\{ \ket{\sigma}_i\}$ is an orthonormal basis of $\mathcal{H}_{\textup{sys}}$.
All that is left to do is to demonstrate that for any given value of $i$ and $\pm$, $\{\ket{\rho(\phi_{i,\pm}-\frac{m}{2^n})} \}_m$ forms an orthogonal basis for $\mathcal{H}_{\textup{iph}}$.
For $m,m'\in\mathbb{Z}^+\in[0,2^n-1]$ and all $\phi \in \mathbb{R}$ we have $\braket{\rho(\phi - m/2^n)}{\rho(\phi - m'/2^n)} = \delta_{m,m'}$ since, 
\begin{subequations}
\begin{align}
    \braket{\rho(\phi - m/2^n)}{\rho(\phi - m'/2^n)} &= \sum_{a=0}^{2^n-1} f^*(\phi-\frac{m}{2^n},a)f(\phi-\frac{m'}{2^n},a)\\
&= \frac{1}{2^n}\sum_{a,y,z = 0}^{2^n-1}W_a(y)W^*_a(z)\e^{2\pi \ii [y(\phi-m/2^n)-z(\phi-m'/2^n)]}\\
& = \frac{1}{2^{n}}\sum_{y,z= 0}^{2^n-1}\delta_{y,z}\e^{2\pi \ii [y(\phi-m/2^n)-z(\phi-m'/2^n)]}\label{eqn normalised window}\\
& = \frac{1}{2^{n}}\sum_{y= 0}^{2^n-1}\e^{2\pi \ii y(m'-m)/2^n} = \delta_{m,m'},
\end{align}
\end{subequations}
where \cref{eqn normalised window} follows from the unitary condition on the window function preparation $\sum_a W_a(x)W^*_a(x') = \delta _{x,x'}$.
\end{proof}

\begin{lemma}\label{lm eigenval of X}
Given $x_i\in\mathbb{R}$ the matrix, $X$, has only two non-trivial eigenvalues $\lambda_\pm$:
\begin{equation*}
    X := \begin{pmatrix}0&x_1&\dots&x_{T-1}\\
x_1 &0&\dots&0\\
\vdots&\vdots&\ddots&\vdots\\
x_{T-1}&0&\dots&0
 \end{pmatrix} \quad\text{and}\quad \lambda_{\pm} = \pm \sqrt{\sum_{i=1}^{T-1}\abs{x_i}^2}.
\end{equation*}
\end{lemma}
\begin{proof}
Consider the characteristic polynomial: $\det(\lambda \mathbb{I} - X) = \lambda^T - \lambda^{T-2}\sum_{i=1}^{T-1}\abs{x_i}^2 = 0$,
solving for $\lambda$ gives, 
\begin{align}
    \lambda^T & = \lambda^{T-2}\sum_{i=1}^{T-1}\abs{x_i}^2\\
    \lambda^2 & = \sum_{i=1}^{T-1}\abs{x_i}^2,
\end{align}
which has only two solutions stated in the lemma. 
\end{proof}

We finally require the well known result:
\begin{lemma}\label{lm Weyl}\textup{(Spectral stability)}
Given two Hermitian matrices $N$ and $M$ with ordered eigenvalues $\nu_1\geq\nu_2\geq\dots\geq\nu_n$ and $\mu_1\geq\mu_2\geq\dots\geq\mu_n$ respectively.
For all $k$
\[\min_j \abs{\nu_j - \mu_k} \leq \norm{N-M}_2,\]
where $\norm{A}_2 = \sqrt{\lambda_\textup{max}(A^*A)}$. 
\end{lemma}
\noindent See e.g. \cite{tao} for a pedagogical proof of \cref{lm Weyl}.
We are now in the position to prove \cref{lm U spectra}.

\begin{proof} (of \cref{lm U spectra})
Consider the projectors corresponding to the two described controlled reflections: $P':= \frac{1}{2}(\mathds{1} - \textup{c}\mbox{-}R_{\tilde{\psi}_{0,+}})$ and $Q':= \frac{1}{2}(\mathds{1} - \textup{c}\mbox{-}\mathcal{B}[\hat{F}])$.
Let $t$ be the number of qubits in the system register $\mathcal{H}_\text{sys}$.
\cref{lm iterate spectrum} gives that the eigenspectra of $\mathcal{U}$ is can be analysed via the singular values of the projector product. 
Following the same method as \cref{thm exact estimation}, the corresponding projector product squared is given by, 
\begin{subequations}
\begin{align}\label{eqn full pp}
    P'\cdot Q' \cdot P' & =\sum_{i,j=0}^{2^t-1}\sum_{\pm,\pm'} \alpha_{i,j}^{\pm,\pm'}\ket{\psi_{i,\pm}}\bra{\psi_{j,\pm'}}_{\text{sys}, \mathcal{B}_H}\otimes \ket{\rho(\phi_{i,\pm}- \tilde{\phi}_{0,+})}\bra{\rho(\phi_{j,\pm'} - \tilde{\phi}_{0,+})}_{\text{iph}} \otimes \ket{0}\bra{0}_{\mathcal{B}_F}\\
    &\stackrel{!}{=} \sum_{k}\abs{\omega_k}^2 \ket{a_k}\bra{a_k}_{\text{sys, iph, }\mathcal{B}_F,\mathcal{B}_H}\label{eqn full pp 2}
\end{align}
\end{subequations}
where $\alpha_{i,j}^{\pm,\pm'} = \frac{1}{4}\bra{\sigma_i}\mathds{1}-\hat{F}\ket{\sigma_j}\braket{\rho(\phi_{i,\pm}-\tilde{\phi}_{0,+})}{0}\braket{0}{\rho(\phi_{j,\pm'}-\tilde{\phi}_{0,+})}$. 
While \cref{eqn full pp} is not diagonal, it must be brought into the diagonal form in \cref{eqn full pp 2}.
It is left to solve for $\omega_k$ and $\ket{a_k}$.

If the algorithm is working functionally then $|\braket{\rho(\phi_{j\neq 0,\pm} - \tilde{\phi}_{0,+})}{0}_{\text{iph}}|$ should be small (corresponding to the reflection not being contaminated too much by excited states). 
Note also that $\{\ket{\psi_{i,\pm}}\ket{\rho(\phi_{i,\pm}-m)}\}_{i,\pm,m}$ is a complete orthogonal basis for $\mathcal{H}_{\text{sys}}\otimes \mathcal{H}_{\text{ph}}\otimes \mathcal{H}_{\mathcal{B}_H}$ (\cref{claim orthog}).
Therefore we look to calculate $\omega_k$ relative to the decoupled matrix:
 \begin{equation}
 \tilde{M} := \sum_{i,j=0}^{2^t-1}\sum_{\pm,\pm'} \beta_{i,j}^{\pm,\pm'}\ket{\psi_{i,\pm}}\bra{\psi_{j,\pm'}}_{\text{sys}, \mathcal{B}_H} \otimes
 \ket{\rho(\phi_{i,\pm}- \tilde{\phi}_{0,+})}\bra{\rho(\phi_{j,\pm'} - \tilde{\phi}_{0,+})}_{\text{iph}} \otimes \ket{0}\bra{0}_{\mathcal{B}_F}
 \end{equation}
where 
\begin{equation*}
    \beta_{i,j}^{\pm,\pm'} = \begin{cases} 0 & \text{if $(i=0 \lor j = 0) \land (\pm = - \lor \pm' = -)$}\\
    \alpha_{i,j}^{\pm,\pm'} & \text{otherwise}
    \end{cases}
\end{equation*}
Since $\tilde{M}$ is decoupled, it is simple to identify the eigenvector $\ket{\psi_{0,+}}_{\text{sys}, \mathcal{B}_H}\ket{\rho(\phi_{0,+}-\tilde{\phi}_{0,+})}_{\text{iph}}\ket{0}_{\mathcal{B}_F}$ with eigenvalue $\frac{1}{4}\bra{\sigma_0}\mathds{1}-\hat{F}\ket{\sigma_0}|\braket{\rho(\phi_{0,+}-\tilde{\phi}_{0,+})}{0}|^2$.

The difference $\Delta\otimes \ket{0}\bra{0}_{\mathcal{B}_F}:= P'\cdot Q'\cdot P'-\tilde{M}$ is a real\footnote{Real due to the phase factors absorbed to the basis.} symmetric matrix, in the basis $\{\ket{\psi_{i,\pm}}\ket{\rho(\phi_{i,\pm}-m)}\}_{i,\pm,m}$:
 \begin{equation}
\Delta = \begin{pmatrix}0&F_{0,0}\braket{\rho_{0}^+}{0}\braket{0}{
\rho_{0}^-}&\dots&F_{0,T-1}\braket{\rho_{0}^+}{0}\braket{0}{
\rho_{T-1}^-}\\
F_{0,0}\braket{\rho_{0}^-}{0}\braket{0}{
\rho_{0}^+}&0&\dots&0\\
F_{1,0}\braket{\rho_{1}^+}{0}\braket{0}{
\rho_{0}^+}&\vdots&\ddots&\vdots\\
\vdots& & &\\
F_{T-1,0}\braket{\rho_{T-1}^+}{0}\braket{0}{
\rho_{0}^+}&\vdots&&\vdots\\
F_{T-1,0}\braket{\rho_{T-1}^-}{0}\braket{0}{
\rho_{0}^+}&0&\dots&0
 \end{pmatrix},
 \end{equation}
 where for conciseness we introduce the compressed notation $F_{i,j} := \frac{1}{4}\bra{\sigma_i}\mathds{1} - \hat{F} \ket{\sigma_j}$ and $\ket{\rho_{i}^{\pm}} := \ket{\rho(\phi_{i,\pm}-\tilde{\phi}_{i,+})}$.
Due to the structure of this matrix the only non-zero eigenvalues as per \cref{lm eigenval of X} are given by
\begin{equation}\label{eqn lambda pm}
\pm \sqrt{\abs{F_{00} \braket{\rho_0^-}{0}\braket{0}{\rho_0^+}}^2 + \sum_{\sigma\in\pm}\sum_{i=1}^{T-1}\abs{F_{0,i}\braket{\rho_0^+}{0}\braket{0}{\rho_i^\sigma}}^2 }.
\end{equation}

The maximum difference between an eigenvalue of $\tilde{M}$ and the closest eigenvalue of $P' \cdot Q' \cdot P'$ is upper bounded by the operator norm of the difference (\cref{lm Weyl}) therefore, 
\begin{subequations}
\begin{align}
    \min_{j} \abs{\lambda_j(P' \cdot Q' \cdot P') - F_{00}\abs{\braket{\rho_0^+}{0}}^2} & \leq \norm{\Delta}_{\text{op}}\\
    & = \sqrt{\abs{F_{00} \braket{\rho_0^-}{0}\braket{0}{\rho_0^+}}^2 + \sum_{\pm}\sum_{i=1}^{T-1}\abs{F_{0,i}\braket{\rho_0^+}{0}\braket{0}{\rho_i^\pm}}^2 }\\
    & \leq 2 |\braket{\rho_0^+}{0}|\max_{j>0\in[0,2^t-1],\pm} \left\{|\braket{0}{\rho_j^\pm}|,|\braket{0}{\rho_0^-}| \right\}\sum_{i = 0}^{2^t -1}\abs{F_{0,i}}\\
    & \leq 2\abs{\braket{\rho_0^+}{0}} \max_{j>0\in[0,2^t-1],\pm} \left\{|\braket{0}{{\rho_j^\pm}}|,|\braket{0}{\rho_0^-}| \right\}\label{eqn final lm 3}
\end{align}
\end{subequations}
The final line follows since $\sum_{i = 0}^{2^t -1}\abs{F_{0,i}} \leq \max_k\sum_{i = 0}^{2^t -1}\abs{F_{k,i}}$ is upper bounded by the infinity matrix norm which is sub-multiplicative. 
Since $\hat{F}$ is a bounded Hermitian matrix there exists a unitary $U$ such that $\hat{F} = U^\dagger D U$ where $D$ is a diagonal bounded matrix. 
Therefore $\norm{\hat{F}}_{\text{op}} \leq \norm{U^\dagger}_{\text{op}} \norm{D}_{\text{op}} \norm{U}_{\text{op}} \leq 1$, where $\norm{D}_{\text{op}} \leq 1$ and $\norm{U}_{\text{op}} =1$.

Returning to the result statement, \cref{lm iterate spectrum} gives that the unitary $\mathcal{U}$ has eigenphases $\theta_{0,\pm}=\pm \frac{1}{2\pi} \cos^{-1} \left(2\abs{\omega}^2 -1 \right)$.
The singular values $\omega$ correspond to the absolute squared values of the eigenvalues of the projector product $\lambda(P'\cdot Q' \cdot P') = \abs{\omega}^2$.
Hence, \cref{eqn final lm 3} leads directly to the result statement:
\begin{subequations}
    \begin{align}
        \abs{\omega}^2 & = F_{00}|\braket{\rho_0^+}{0}|^2 \pm 2\abs{\braket{\rho_0^+}{0}} \max_{j>0\in[0,2^t-1],\pm} \left\{\abs{\braket{0}{{\rho_j^\pm}}},\abs{\braket{0}{\rho_0^-}} \right\}\\
        & = \frac{1}{4}\bra{\sigma_0}\mathds{1} - \hat{F} \ket{\sigma_0}|\braket{\rho_0^+}{0}|^2 \pm 2\abs{\braket{\rho_0^+}{0}} \max_{j>0\in[0,2^t-1],\pm} \left\{\abs{\braket{0}{{\rho_j^\pm}}},\abs{\braket{0}{\rho_0^-}} \right\}.
    \end{align}
\end{subequations}
\end{proof}

\subsection{Proof of error analysis}\label{appen:error}
\properror*
\begin{proof}

Recapping the three contributions of error:
\begin{enumerate}
    \item \textbf{Reflection error} $\epsilon_1 :=  \abs{\abs{\omega}^2 - \mathfrak{f}(\bra{\sigma_0}\hat{F}\ket{\sigma_0}) }$, due to the eigenvalue of the projector product deviating from the decoupled matrix, where $\mathfrak{f}(X) = \frac{|\braket{\rho_0^+}{0}|^2}{4} (1 - X)$.
    From \cref{lm U spectra}, $\epsilon_1 \leq \abs{\braket{\rho_0^+}{0}}2 \max_{j>0,\pm} \left\{|\braket{0}{\rho_j^\pm}|,|\braket{0}{\rho_0^-}| \right\}$.
    \item \textbf{Pre-learning error}, of the estimate $|\braket{\rho_0^+}{0}|^2\pm \delta(|\braket{\rho_0^+}{0}|^2)$. Recall that for the first step of the algorithm is to run the circuit with the identity observable to obtain as estimate of this value that will be used to compute the final expectation value from the obtained measurement. The error in this estimate is dominated by the bit discretisation of the oQPE. The oQPE estimates $\theta'_{\pm} = \pm \frac{1}{\pi}\cos^{-1}(2|\braket{\rho_0^+}{0}|^2 - 1)$, bit discretisation gives $\delta(\theta') = 2^{-n_o}$ therefore $\delta(|\braket{\rho_0^+}{0}|^2) = \pi\sqrt{|\braket{\rho_0^+}{0}|^2(1-|\braket{\rho_0^+}{0}|^2)}\delta(\theta') \leq \pi 2^{-(n_o+1)}$.
    \item \textbf{oQPE error} from imperfect estimation of the iterate eigenphase due to bit discretisation in the outer QPE, $ \abs{\theta_{0,\pm} - \tilde{\theta}_{0,\pm}}\leq 2^{-(n_o+1)}$.
\end{enumerate}

\cref{lm iterate spectrum} shows that the eigenphase of the iterate $\mathcal{U}$ is related to the squared eigenvalues of the projector product by, $\theta_\pm = \pm \frac{1}{2\pi}\cos^{-1}(2\omega^2 - 1)$.
Propagation of errors gives $\frac{d\theta}{d(\omega^2)} = - \frac{1}{2\pi\sqrt{\omega^2(1-\omega^2)}}$ we can relate error in $\theta$ to error in the squared eigenvalue,
\begin{equation}
   \delta(\omega^2) = \frac{d(\omega^2)}{d\theta}\delta \theta = 2\pi \sqrt{\omega^2(1-\omega^2)}\delta \theta,
\end{equation}
where $\delta(\cdot)$ denotes the absolute error.
Therefore, $\delta(\mathfrak{f}(\bra{\sigma_0}\hat{F}\ket{\sigma_0})) = \delta(|\omega|^2) + \epsilon_1$.
Relate this to the error in $\bra{\sigma_0}\hat{F}\ket{\sigma_0}$, by introducing the notation $\langle F\rangle : = \bra{\sigma_0}\hat{F}\ket{\sigma_0}$, $p =|\braket{\rho_0^+}{0}|^2$ and defining $g(X):= \frac{\mathfrak{f}(X)}{p} = \frac{1}{4}(1-X)$:
\begin{subequations}
\begin{align}
    \left(\frac{\delta(g(\langle F\rangle))}{g(\langle F\rangle)} \right)^2 &= \left(\frac{\delta(\mathfrak{f}(\langle F\rangle))}{\mathfrak{f}(\langle F\rangle)} \right)^2 + \left(\frac{\delta(p)}{p} \right)^2\\
    \delta(g(\langle F\rangle))^2 & = \frac{\delta(\mathfrak{f}(\langle F\rangle))^2}{p^2} + \frac{g(\langle F\rangle)^2\delta(p)^2}{p^2}\\
    \delta(g(\langle F\rangle))^2 & \leq  \frac{\delta(\mathfrak{f}(\langle F\rangle))^2}{p^2} + \frac{\delta(p)^2}{4p^2}\\
    \delta(g(\langle F\rangle)) & \leq  \frac{1}{p}\sqrt{\delta(\mathfrak{f}(\langle F\rangle))^2 + \delta(p)^2/4}.
\end{align}
\end{subequations}
Since $\langle F\rangle \in [-1,1]$ we can upper bound $g(\langle F\rangle) \leq \frac{1}{2}$.
Then we have, 
\begin{subequations}
\begin{align}
    \delta(\langle F\rangle) & = 4 \delta(g(\langle F\rangle))\\
    & \leq \frac{4}{p}\sqrt{\delta(\mathfrak{f}(\langle F\rangle))^2 + \delta(p)^2/4}.
\end{align}
\end{subequations}

Combining all these components we have, 
\begin{subequations}
\begin{align}
    \abs{\bra{\sigma_0}\hat{F}\ket{\sigma_0} - F_\text{est}} &=\delta(\langle F\rangle)\\
        & \leq \frac{4}{p}\sqrt{\delta(\mathfrak{f}(\langle F\rangle))^2 + \delta(p)^2/4}\\
        &  \leq \frac{4}{p}\sqrt{[\delta(|\omega|^2) + \epsilon_1]^2 + \left[\frac{\pi 2^{-n_0}}{2}\right]^2}\\
        &  \leq \frac{4[\delta(|\omega|^2) + \epsilon_1]}{p} + \frac{\pi 2^{-n_o}}{p}\\
        & \leq \frac{4}{p} \left[2\pi 2^{-(n_0+1)}\sqrt{\omega^2(1-\omega^2)}+\epsilon_1 \right] + \frac{\pi2^{-n_o}}{p}\\
        & \leq \frac{3\pi 2^{-n_o}}{p} + \frac{4\epsilon_1}{p}.\label{eqn:errorsindepth}
        \end{align}
\end{subequations}
Finally substituting in for $\epsilon_1$ and $p :=|\braket{\rho_0^+}{0}|^2$ gives,
\begin{equation}
    \epsilon \leq \frac{3\pi 2^{-n_o}+ 8|\braket{\rho_0^+}{0}|\max_{j>0,\pm} \left\{|\braket{0}{\rho_j^\pm}|,|\braket{0}{\rho_0^-}| \right\} + c }{|\braket{\rho_0^+}{0}|^2},
\end{equation}
where $c$ is any non-negative constant.
Since this loosens the bound we are free to set this constant.
\end{proof}

\subsection{Proof of success probability}

To prove the success probability we use the following well known result about perturbation of eigenvectors. 
\begin{restatable}{thm}{thmDavis}\textup{(Davis-Kahan \cite{eigenvect})}\label{thm Davis-Kahan}
Let $A,B\in \mathbb{C}^{T\times T}$ be Hermitian.
For $i\in\{1,...,T\}$ given that
\[\delta := \min_{j\neq i}\abs{\lambda_{i}(A) - \lambda_j(A)}>0,\]
then
\[\min_{s\in\{\pm1\}} \norm{v_i(A)-sv_i(B)}^2 \leq \frac{8\norm{A-B}_2^2}{\delta}.\]
\end{restatable}
\noindent For a proof of \cref{thm Davis-Kahan} see \cite{eigenvect}.

\propsuccess*

\begin{proof}
In correspondence with \cref{prop error} the probability of success is defined as obtaining an estimate to the expectation value that is accurate to within the error tolerance given in the error bound.
Practically this amounts to a success probability defined as $p_{\textup{success}} := |\bra{\psi_{0,+}}\bra{\rho_0^+}\bra{0} \sum_{k\in \mathcal{K}}\ket{a_k}|^2$ where $\{\ket{a_k} \}_{k\in\mathcal{K}}$ is the set of eigenvectors of the projector product of the iterate $\mathcal{U}$ with eigenvalues $|\lambda_k - F_{00}|\braket{\rho_0^+}{0}|^2|< 2 \max_{j>0,\pm} \left\{|\braket{0}{\rho_j^\pm}|,|\braket{0}{\rho_0^-}| \right\}$.

As in \cref{lm iterate spectrum} we express the eigenvectors of $P'\cdot  Q' \cdot P'$ as a perturbation from those of the decoupled matrix $\tilde{M}$.
\begin{align}
    P'\cdot  Q' \cdot P' & =\sum_{i,j=0}^{2^t-1}\sum_{\pm,\pm'} \alpha_{i,j}^{\pm,\pm'}\ket{\psi_{i,\pm}}\bra{\psi_{j,\pm'}}_{\text{sys}, \mathcal{B}_H}\otimes
    \ket{\rho(\phi_{i,\pm}- \tilde{\phi}_{0,+})}\bra{\rho(\phi_{j,\pm'} - \tilde{\phi}_{0,+})}_{\text{iph}}\otimes \ket{0}\bra{0}_{\mathcal{B}_F}
\end{align}
where $\alpha_{i,j}^{\pm,\pm'} = \frac{1}{4}\bra{\sigma_i}\mathds{1}-\hat{F}\ket{\sigma_j}\braket{\rho(\phi_{i,\pm}-\tilde{\phi}_{0,+})}{0}\braket{0}{\rho(\phi_{j,\pm'}-\tilde{\phi}_{0,+})}$ and
 \begin{equation}
 \tilde{M} := \sum_{i,j=0}^{2^t-1}\sum_{\pm,\pm'} \beta_{i,j}^{\pm,\pm'}\ket{\psi_{i,\pm}}\bra{\psi_{j,\pm'}}_{\text{sys}, \mathcal{B}_H} \otimes
 \ket{\rho(\phi_{i,\pm}- \tilde{\phi}_{0,+})}\bra{\rho(\phi_{j,\pm'} - \tilde{\phi}_{0,+})}_{\text{iph}} \otimes \ket{0}\bra{0}_{\mathcal{B}_F}
 \end{equation}
where 
\begin{equation*}
    \beta_{i,j}^{\pm,\pm'} = \begin{cases} 0 & \text{if $(i=0 \lor j = 0) \land (\pm = - \lor \pm' = -)$}\\
    \alpha_{i,j}^{\pm,\pm'} & \text{otherwise}.
    \end{cases}
\end{equation*}
To prove statements about the perturbation of eigenvectors requires knowledge of the eigengap of the matrix, $\delta$ as described in \cref{thm Davis-Kahan}.
Substitute $\tilde{M} = A$ and $P'\cdot  Q' \cdot P' = B$ into \cref{thm Davis-Kahan} and look at $i$ such that $v_i(\tilde{M}) = \ket{\psi_{0,+}}\ket{\rho_0^+}\ket{0}$ and $\lambda_i(\tilde{M}) = F_{00}|\braket{\rho_0^+}{0}|^2$ (recall the compressed notation: $F_{i,j} := \frac{1}{4}\bra{\sigma_i}\mathds{1} - \hat{F} \ket{\sigma_j}$ and $\ket{\rho_{i}^{\pm}} := \ket{\rho(\phi_{i,\pm}-\tilde{\phi}_{0,+})}$).
Using the bound on the operator norm of $\Delta\otimes \ket{0}\bra{0}_{\mathcal{B}_F}:= P'\cdot  Q' \cdot P'-\tilde{M}$ from \cref{lm U spectra}:
\begin{subequations}
\begin{align}
    \min_{s\in\{\pm1\}} \norm{\ket{\psi_{0,+}}\ket{\rho_0^+}\ket{0}-s v_i( P'\cdot  Q' \cdot P')}^2 &\leq \frac{8\norm{\Delta}_{\text{op}}^2}{\delta}\\
    & \leq \frac{32|\braket{\rho_0^+}{0}|^2\max_{j>0,\pm} \left\{|\braket{0}{\rho_j^\pm}|,|\braket{0}{\rho_0^-}| \right\}^2}{\delta}\label{eqn sub in delta}.
\end{align} 
\end{subequations}
The eigengap $\delta$ is the gap between the eigenvalue $\lambda_i(\tilde{M}) = F_{00}|\braket{\rho_0^+}{0}|^2$ and the next closest eigenvalue of $\tilde{M}$. 
Due to the decoupled nature of $\tilde{M}$ the other eigenvalues correspond to a sub-matrix without support on the $v_i(\tilde{M}) = \ket{\psi_{0,+}}\ket{\rho_0^+}\ket{0}$ basis.
Recycling the argument from \cref{lm U spectra} we know that the operator norm of this sub-matrix is small.

We then separate into two cases and analyse these separately,
\begin{enumerate}[\textbf{Case} 1:]
    \item $\lambda_i(\tilde{M}) = F_{00}|\braket{\rho_0^+}{0}|^2$ is the largest eigenvalue of $\tilde{M}$ and is lower bounded by \[\lambda_i(\tilde{M})\geq 2 |\braket{\rho_0^+}{0}|\max_{j>0,j\pm}\left\{|\braket{0}{\rho_j^\pm}|,|\braket{0}{\rho_0^-}| \right\} + c\]
    \item Otherwise $\lambda_i(\tilde{M})$ is itself small: $\lambda_i(\tilde{M})< 2|\braket{\rho_0^+}{0}|\max_{j>0,j\pm}\left\{|\braket{0}{\rho_j^\pm}|,|\braket{0}{\rho_0^-}| \right\} + c$
\end{enumerate}
for some non-zero constant $c$.
These two cases clearly encapsulate all possible scenarios for any value of $c$, in practice we set $c$ bounded away from 0 so that if the eigengap is closing we do not use \cref{thm Davis-Kahan} to bound the success probability and instead consider that all eigenvectors lie within the error tolerance.

\paragraph{Case 1 analysis:} $\lambda_i$ is large so that the closest eigenvalue to $\lambda_i$ is now the maximum eigenvalue of the sub-matrix. 
Since $\abs{F_{i,j}}\leq 1/2$ we can upper bound the maximum eigenvalue of the sub-matrix by \\$2 |\braket{\rho_0^+}{0}|\max_{j>0,\pm} \left\{|\braket{0}{\rho_j^\pm}|,|\braket{0}{\rho_0^-}| \right\}$. 
Hence the eigengap is lower bounded by $\delta \geq c$ -- note that here if we consider $c$ going to zero the success probability from \cref{thm Davis-Kahan} vanishes.
Substituting this into \cref{eqn sub in delta} gives
\begin{equation}
    \min_{s\in\{\pm1\}} \norm{\ket{\psi_{0,+}}\ket{\rho_0^+}\ket{0}-s v_i(P'\cdot  Q' \cdot P')}^2\leq \frac{32 |\braket{\rho_0^+}{0}|^2\max_{j>0,\pm} \left\{|\braket{0}{\rho_j^\pm}|,|\braket{0}{\rho_0^-}| \right\}^2}{c}
\end{equation}
The probability of success is then given by $p_{\text{success}} \geq |\braket{0}{\rho_0^+}|^2 \left(1 - \frac{32 \max_{j>0,\pm} \left\{|\braket{0}{\rho_j^\pm}|,|\braket{0}{\rho_0^-}| \right\}^2}{c}\right)$ since the state prepared is $\ket{\psi_{0,+}}_{\text{sys}, \mathcal{B}_H}\ket{0}_{\text{iph}}\ket{0}_{\mathcal{B}_F}$.

\paragraph{Case 2 analysis:} Here the eigenvalue of interest, $\lambda_i$, is itself small, and hence the eigengap closes as $c$ decreases, allowing the eigenvectors to mix and the overlap between the initial state and the eigenvector corresponding to $\lambda_i$ will decrease.
However, this can only happen since all the eigenvalues are now close, so while the algorithm may estimate an eigenvalue corresponding to a different eigenvector this is an $(2 |\braket{\rho_0^+}{0}|\max_{j>0,\pm} \{|\braket{0}{\rho_j^\pm}|,|\braket{0}{\rho_0^-}|\} + c)$-approximation of the desired eigenvalue with probability 1.
For consistency this does however increase the error for the estimate which is why $4c/p$ appears as an additional error in \cref{prop error}.

Since all outcomes in case 2 results in an estimate of the expectation value within the error tolerance of \cref{prop error} the total success probability is given by the expression for case 1.
The additional error in the eigenvalue $c$ must be taken to be a non-negative value and coincides with the extra tolerance added to the error bound in \cref{prop error}.
\end{proof}
\pagebreak

\section{Resource estimation details}\label{appen RE}

This appendix details how the numerics quoted in \cref{fig RE numeric} and \cref{table logical qubits} were generated. 
We first give an overview, and then give specifics on the use of BLISS-THC and the error model used to do the truncation as well as the corresponding data.

\subsection{Overview}
\paragraph{Overall Toffoli cost:} \cref{fig RE numeric} gives a quantum resource estimates for the number of Toffoli gates required to learn a physical observable (kinetic energy, $x$-dipole moment and ERI) with respect to different molecular systems. 
Recall some notation from the algorithm: the arbitrary state preparation unitary for the initial state of the system $\textup{ASP}$; the unitary that prepares the window state $\hat{W}$; the qubitised encoding of the Hamiltonian $\mathcal{Q}[H]$; a quantum Fourier transform on $x$ qubits $\textup{QFT}_x$; the self-inverse encoding of the observable $\mathcal{B}[F]$; the number of phase qubits in the outer QPE $n_0$ and the number of phase qubits in the inner QPE $n$.
From examining the circuit diagram in \cref{fig circuit}, assuming a simple additive gate count model, the approximate cost of one run of the full algorithm is given by:
\begin{equation}\label{eqn cost}
  \mathcal{T}(\textup{ASP}) + 2^{n_0} \mathcal{T}(\hat{W}) + 2^{n_0+n-1} \mathcal{T}(\mathcal{Q} [H]) + 2^{n_0+1}\mathcal{T}(\textup{QFT}_n) + 2^{n_0-1}\mathcal{T}(\mathcal{B}[F]) + \mathcal{T}(\textup{QFT}_{n_0}),
\end{equation}
where $\mathcal{T}(\cdot)$ denotes the Toffoli cost associated with that operation. 

In textbook QPE a phase register of $n'$ qubits requires $\sum_{r=0}^{n'-1}2^r = 2^{n'} -1$ calls the controlled unitary, one call to a phase state preparation unitary and one to a $\text{QFT}_{n'}$.
However, by controlling on both subspaces of each phase register qubit and applying both $U$ and $U^\dagger$ (see for example \cite{mcclean_kickback}[Figure 2]) the number of unitary calls can be squeezed by roughly\footnote{Additionally if the unitary is comprised of a product of reflections -- as is the case for both QPE subroutines in this algorithm -- further reductions can be made by not controlling some of the reflections.} a factor of 2.
Therefore, given $n_0$ outer phase qubits, $\mathcal{T}(\text{oQPE}) = 2^{n_0-1}\mathcal{T}(\mathcal{U}) + \mathcal{T}(\text{QFT}_{n_0})$.
The walk operator can be further decomposed as $\mathcal{T}(\mathcal{U})= 2 \mathcal{T}(\text{iQPE}) + \mathcal{T}(\text{Refl}_\text{iph}) + \mathcal{T}(\textup{c}\mbox{-}\mathcal{B}[F])$.
Finally the inner QPE with $n$ total phase qubits requires $\mathcal{T}(\text{iQPE}) = 2^{n-1}\mathcal{T}(\mathcal{Q}[H]) + \mathcal{T}(\text{QFT}_n) + \mathcal{T}(\hat{W})$.
Putting this together yields the expression for the total cost in \cref{eqn cost}.

This expression highlights that the leading cost in the algorithm is the block encoding of the Hamiltonian due to the large number of repetitions required. 
We find that the Toffoli cost of the other components is negligible compared to the block encodings, therefore only the cost of the Hamiltonian and observable encodings where included in \cref{fig RE numeric}. 

\paragraph{Overall logical qubit cost:} 
The logical qubit estimates given in \cref{table logical qubits} correspond to the qubit highwater -- the maximum number of logical qubits required to be in coherent superposition at any given point in the circuit. 
Subroutines will require temporary ancillae that are released at the end of the routine by uncomputes.
This is opposed to persistent ancillae, which cannot be reused freely, as they are either are referenced by gates in later subroutines or remain entangled.
Not reusing the temporary ancillae across subroutines can greatly overestimate the resources required to run the algorithm and so this distinction between temporary and persistent is essential. 
Not accounting for reusing the temporary ancilla across routines makes for the significantly higher logical qubit counts reported in previous estimates \cite{seve}.
In addition to the temporary ancillae the logical qubit count includes all registers that are pervasive through the algorithm (those that appear explicitly in \cref{fig circuit}).
Qubit highwater additions due to state preparation and QSP implementation are neglected. 

\paragraph{Error budgeting:} As part of the problem instance we set an absolute error tolerance $\epsilon_\text{target}$ for the expectation value.
These are chosen to be reasonable with respect to literature: $\epsilon_{\text{kin.}} = 1.6 \text{ mHa}$, $\epsilon_{\text{dip.}} = 10\text{ mDebye}$, and $\epsilon_{\text{eri}} = 1.6 \text{ mHa}$. For the numerics we assume perfect gate implementation and state preparation, however due to the limited qubit registers the algorithm will only estimate the expectation value approximately, as discussed at length in \cref{appen:error}.
The error, $\epsilon_\text{inner}$ is the contribution to the final estimate from the inner QPE ($\epsilon_\text{inner} = \frac{4\epsilon_1}{p}$ from \cref{eqn:errorsindepth}) and $\epsilon_\text{inner}/\lambda_F$ is plotted in \cref{fig RE} for various window functions.
The outer error, $\epsilon_\text{outer}$, is the contribution to the final estimate from the inner QPE ($\epsilon_\text{outer} = \frac{3\pi 2^{-n_o}}{p}$ from \cref{eqn:errorsindepth}) and is exponentially suppressed by increasing the size of the outer phase register.

Furthermore, the dominant contribution to the resource estimation is the cost of the block encoding.
The block encoding implementation introduces another error from data truncation.
To obtain numerical estimates, this error must be budgeted for and so we introduce $\epsilon_\text{data}$ as the absolute error contribution due to the chosen implementation of the block encoding.
Hence, we require a total error target decomposition,
\begin{equation}
    \epsilon_\text{target} = \epsilon_\text{inner} + \epsilon_\text{outer} + \epsilon_\text{data}.
\end{equation}
For the numerics quoted in \cref{fig RE}, a simple error distribution of $\epsilon_\text{inner} = \epsilon_\text{outer} = \epsilon_\text{data} = \epsilon_\text{target}/3$ was chosen.
Some preliminary optimisation was done between $\epsilon_\text{inner}$ and $\epsilon_\text{outer}$ and it was found that although it varied between instance, and equal split was close to optimal. 
We leave further optimisation through refinement of the error model to future work. 

\paragraph{Number of inner phase qubits, $\boldsymbol{n=l+m}$:} The \say{baseline} inner QPE number of qubits $l$ is chosen such that the ground state of the Hamiltonian can be distinguished from the first excited state. 
This requires an estimate for the Hamiltonian one-norm $\lambda_H$ and an upper bound on the spectral gap $\Delta$.
The relationship between the number of qubits required to distinguish the ground state of the qubitised Hamiltonian is given by
\begin{equation}\label{eqn baseline qubits}
    l = \Bigl\lceil\log_2\left(\frac{2\pi \lambda_H}{\Delta}\right)\Bigr\rceil.
\end{equation}
Recalling the relation between the eigenphase of a qubitised operator and its spectra we require
\begin{equation*}
    \frac{1}{2\pi} \abs{\cos^{-1}\left(\frac{\lambda_0}{\lambda_H} \right) - \cos^{-1}\left(\frac{\lambda_0 + \Delta}{\lambda_H}\right)}> \frac{1}{2^l},
\end{equation*}
where $\lambda_0$ is the ground state energy of the system.
The left-hand side is minimised by $\lambda_0 = 0$, we then expand $\cos^{-1}(x) \approx \frac{\pi}{2} -x$ assuming the relative spectral gap is small.
 
In both the standard rectangular QPE and the Kaiser-windowed version the success amplitude is improved by adding additional phase qubits $m$ -- see \cref{fig RE}. 
The minimum number of additional qubits $m$ is chosen such that the relative error budget $\epsilon_\text{inner}/\lambda_F$, is below the threshold.

\paragraph{Number of outer phase qubits $\boldsymbol{n_o}$: }The number of outer QPE phase qubits is given by
\begin{equation}\label{eqn outer cost}
    n_o = \Bigl\lceil\log_2\left(\frac{3\pi \lambda_F }{p \epsilon_\text{outer}}\right)\Bigr\rceil,
\end{equation}
where we recall that $p=|\braket{\rho_0^+}{0}|^2$.
For the rectangular window this is lower bounded by $4/\pi^2$ whereas for the Kaiser window this is calculated from \cref{fg overlap} [Right] for the chosen $\beta$.
Recall that the phase estimated by the outer QPE corresponds to $\theta = \pm \frac{1}{2\pi}\cos^{-1}(2\abs{\omega}^2-1)$ where $\abs{\omega}^2 = \frac{p}{4}(1- \langle F\rangle)$.
The relationship between the phase error and the expectation value error is governed by \cref{prop error}.
The error induced by a finite outer phase register is $2^{-n_o}$, which leads to the expression in \cref{eqn outer cost}.

\paragraph{Estimating problem parameters:} To obtain representative resource estimates for the expectation value algorithm we quote three different observables over four different systems of varying complexity. 
The one-norm and spectral gap of each system Hamiltonian, in addition to the one-norm of all observables, must be estimated to inform the algorithmic requirements such as the number of phase qubits. 
We follow \cite{seve} in defining the three closed-shell systems -- refer to \cite[Table V]{seve} for details of the molecular geometries used for the calculations, where the cc-pVDZ basis set is used throughout all computations.
For the heme (P450) system, we use the active space selected in \cite{cortes2024fault}, and refer to \cite[Appendix I]{cortes2024fault} for details on the system parameters. 
The spectral gap is computed with DMRG at bond dimension $M=1000$; applied to the full all-electron space for ammonia and water, and to a (30e, 30o) active space for p-Benzyne. We refer readers to \cite{seve} and \cite{cortes2024fault} for further details on these calculations.
The values used for the resource estimate are quoted in \cref{table chemistry data}.

\begin{table}[h!]
    \centering
    \begin{tabular}{c*4r*2r*2r}
\hline \hline 
\bf{System} & \multicolumn{2}{c}{\bf{Property}}  \\
 & $N$ &  $\Delta$  \\
\hline
Water & 24 & 0.302 \\
Ammonia & 29  &0.280 \\
P450 heme & 43  &  0.0069\\
p-Benzyne & 104  & 0.114\\
\hline \hline 
\end{tabular}
    \caption{Hamiltonian spectral gap estimate, $\Delta_H$, in Hartree used in the quantum resource estimate. We follow \cite{seve} in defining the systems for a fair comparison. }\label{table chemistry data}
\end{table}

\paragraph{Block encoding of the Hamiltonian and observable:} Both the observable and system Hamiltonian are block encoded following the optimised tensor hypercontraction implementation \cite{blissthc}. Compared to a previous work \cite{Evenmoreefficient}, this implementation provided further reductions based on using fused adders for Givens orbital rotation gates, improved bit precision allocation for rotation angles, as well as other minor improvements. We refer readers to \cite{blissthc} for exact details on the block-encoding implementation.
Furthermore, a single error budget, $\epsilon_\text{data}$, is used for all of the block-encoding data loading, which is estimated classically. This procedure is described in more detail in \cref{appen:error handling} where the ranks and bit precisions used to generate the data in \cref{sect RE} are quoted. 

\paragraph{Subleading costs:}
The block encoding repetitions constitute the major cost and we find the ASP, QFTs and window preparation unitary contributions are negligible in comparison. 
The benefits of using a tapered phase state in the quantum phase estimation, comes at the cost of having to prepare a more complicated state than uniform superposition.
The parameters that effect the efficiency of preparing the Kaiser window are: number of qubits $n$, bandwidth $\beta$ and error $\eta$.
When assessing state preparation routines for a function $\sum_x\ket{g(x)}$, there are two broad camps we consider: methods that require an amplitude oracle $\sum_x\ket{x}\ket{0}\mapsto \sum_x\ket{x}\ket{\tilde{g}(x)}$ \cite{Grover2002, LKS} and those that do not \cite{Mcardle2022}. 
While both require significant classical pre-processing -- the former in classical computation of the amplitudes and the latter in computation of high-order Taylor series -- we will neglect this and focus on the quantum resources. 
The main advantage of not using an amplitude oracle is the lower ancilla cost.
The Hamiltonian block encoding required for the QPE uses many ancillae and therefore using a qubit-expensive but gate efficient state preparation does not impact the qubit highwater.
To investigate the subleading nature  simple Grover-Rudolph type state preparation was investigated\footnote{There is scope to optimise this sub-routine using a space-time trade off due to the ancillae needed for the block encoding, but as it is already a subleading contribution other optimisation should be prioritised.}.

\subsection{Block encoding error handling for expectation value estimation}\label{appen:error handling}

The two-body observables and Hamiltonians are represented with BLISS-THC, therefore we need a rank and finite bit precision approximation for both operators.
The minimum rank and bit precision parameters are chosen such that
\begin{equation}
    \abs{\bra{\sigma_0(\hat{H})} \hat{F} \ket{\sigma_0(\hat{H})}_\text{corr} -\bra{\sigma_0(\tilde{H})} \tilde{F} \ket{\sigma_0(\tilde{H})}_\text{corr}} < \epsilon_\text{data}
    \label{correlation error}
\end{equation}
where $\hat{F}$/$\hat{H}$ are the unnormalised target observable/Hamiltonian operators without truncation and $\tilde{F}/\tilde{H}$ are the unnormalised observable/Hamiltonian operators truncated to the rank and bit precision parameters used in each block-encoding scheme\footnote{Note that elsewhere we use $F$, $H$ with $\norm{F}_\text{op}, \norm{H}_\text{op}\leq 1$, but for this classical prepossessing step we refer to the unnormalised operators.  }.
Similar to previous works \cite{Evenmoreefficient,blissthc}, we choose the minimum rank and bit precision parameters based on the expectation value of the correlation component of the 1- and 2-body Reduced Density Matrices (1- and 2-RDMs).
In other words, it is assumed that the full expectation value can be decomposed as, $\bra{\sigma_0(\hat{H})} \hat{F} \ket{\sigma_0(\hat{H})} = \bra{\sigma_0(\hat{H})} \hat{F} \ket{\sigma_0(\hat{H})}_{\text{HF}} + \bra{\sigma_0(\hat{H})} \hat{F} \ket{\sigma_0(\hat{H})}_{\text{corr}}$, where the first term uses the Hartree-Fock 1- and 2-RDMs, while the second term uses the remaining traceless correlation-component 1- and 2-RDMs.
This implies that the actual observable estimate used in the expectation value algorithm algorithm will be, $\langle\tilde{F}\rangle_\text{est.} = \bra{\sigma_0 (\hat{H})} \hat{F} \ket{\sigma_0 (\hat{H})}_{\text{HF}} + \bra{\sigma_0(\tilde{H})} \tilde{F} \ket{\sigma_0 (\tilde{H})}_{\text{corr}}$, where the approximate Hartree-Fock contribution is subtracted out and replaced with a Hartree-Fock observable estimate based on non-truncated Hamiltonian and observables.
We find that this methodology improves the overall truncation scheme with errors that converge more reliably.

Since FCI-level calculations are not possible beyond small system sizes, we estimate $\bra{\sigma_0(\hat{H})} \hat{F} \ket{\sigma_0(\hat{H})}$ using coupled cluster singles and doubles (CCSD), which provide 1- and 2-RDMs necessary for observable estimation. The truncated expectation value, $\bra{\sigma_0(\tilde{H})} \tilde{F} \ket{\sigma_0(\tilde{H})}$, is also estimated with same CCSD procedure performed with the truncated Hamiltonian and associated observable. In all cases, we are able to achieve $\epsilon_\text{data} = \epsilon_\text{target}/3$ with the ranks, norms and bit precision parameters given in \cref{table:kineticdata,table:dipoledata,table:eridata}.

\begin{table}[h]
\centering
\renewcommand{\arraystretch}{1.2}
\setlength{\tabcolsep}{10pt}
\begin{tabular}{c|ccc|cc|c}
\hline \hline
\textbf{System} & \multicolumn{1}{c}{\makecell{$\lambda_\text{H}$ \\ ($\mathrm{E_h}$)}} & $M_\text{H}$ & $\aleph_\text{H}$, $\beth_\text{H}$ & \multicolumn{1}{c}{\makecell{$\lambda_\text{K}$ \\ ($\mathrm{E_h}$)}} & $\aleph_\text{K}$, $\beth_\text{K}$ & \multicolumn{1}{c}{\makecell{CCSD corr. \\ error ($\mathrm{mE_h}$)}} \\
\hline
\multirow{1}{*}{Water} &  107.49 & 80 & 13  & 69.16 &  9 & 0.445 \\
\hline
\multirow{1}{*}{Ammonia}   & 128.91 & 80 & 14  & 57.25 & 11 & 0.427 \\ \hline
 \multirow{1}{*}{P450 heme}  & 87.34  & 160 & 19 & 58.66 & 13 & 0.437 \\
 \hline
  \multirow{1}{*}{p-Benzyne} &  660.91 & 440 & 16  & 199.25 & 13 & -0.398 \\
\hline\hline
\end{tabular}
\caption{Block encoding details of the \emph{kinetic energy operator} and the Hamiltonian used in for the quantum resource estimate in \cref{fig RE numeric}. 
Since kinetic energy is a one body operator, BLISS-THC is not required. The one-body operator is block-encoded via an eigendecomposition, followed by symmetry-shifting of the eigenvalues to reduce the one-norm.
For each system the following properties are given: $\lambda_\text{H}$ ($\lambda_\text{K}$) the one-norm of the Hamiltonian (kinetic energy operator) in Hartree; $M_\text{H}$ the THC rank of the Hamiltonian; $\aleph_\text{H}$ ($\aleph_\text{K}$)the number of bits of precision used for the `keep' probabilities in alias sampling within the Hamiltonian (kinetic energy operator) block encoding; $\beth_\text{H}$ ($\beth_\text{K}$) the number of bits of precision in the Givens rotation angles required for the SELECT operator of the Hamiltonian (kinetic energy operator); observable correlation error based on coupled cluster singles and doubles (CCSD) reduced density matrices, following the procedure discussed after \cref{correlation error}. For all the observables considered here we take $\aleph_\text{H} = \beth_\text{H}$ and $\aleph_\text{K} = \beth_\text{K}$ to simplify the optimisation.}\label{table:kineticdata}
\end{table}

\begin{table}[h]
\centering
\renewcommand{\arraystretch}{1.2}
\setlength{\tabcolsep}{10pt}
\begin{tabular}{c|ccc|cc|c}
\hline \hline
\textbf{System} & \multicolumn{1}{c}{\makecell{$\lambda_\text{H}$ \\ ($\mathrm{E_h}$)}} & $M_\text{H}$ & $\aleph_\text{H}$, $\beth_\text{H}$ & \multicolumn{1}{c}{\makecell{$\lambda_\text{D}$ \\ ($\mathrm{a.u.}$)}} & $\aleph_\text{D}$, $\beth_\text{D}$ & \multicolumn{1}{c}{\makecell{CCSD corr. \\ error ($\times 10^{-3}\,\mathrm{a.
u.}$)}} \\
\hline
\multirow{1}{*}{Water}  & 105.66 & 60 & 10  & 20.70 & 10 & -0.630 \\
\hline
\multirow{1}{*}{Ammonia}  & 118.66 & 80 & 10  & 28.39 & 10 & 0.608 \\
 \hline
\multirow{1}{*}{P450 heme}  & 84.04  & 140 & 15 & 38.43 & 14 & -0.550 \\ \hline 
 \multirow{1}{*}{p-Benzyne}  & 660.22 & 430 & 17  & 238.59 & 13 & -0.950 \\ \hline
\hline
\end{tabular}
\caption{Block encoding details of the \emph{dipole observable} in the $x$ direction and the Hamiltonian used in for the quantum resource estimate in \cref{fig RE numeric}. Since the dipole moment operator is a one-body operator, BLISS-THC is not required. The one-body operator is block-encoded via an eigendecomposition, followed by symmetry-shifting of the eigenvalues to reduce the 1-norm.
See \cref{table:kineticdata} for a description of the table headings. $\lambda_\text{D}$ is the one norm of the dipole operator in atomic units.}\label{table:dipoledata}
\end{table}

\begin{table}[h]
\centering
\renewcommand{\arraystretch}{1.2}
\setlength{\tabcolsep}{10pt}
\begin{tabular}{c|ccc|ccc|c}
\hline \hline
\textbf{System} &  \multicolumn{1}{c}{\makecell{$\lambda_\text{H}$ \\ ($\mathrm{E_h}$)}}  & $M_\text{H}$ & $\aleph_\text{H}$, $\beth_\text{H}$ & \multicolumn{1}{c}{\makecell{$\lambda_\text{I}$ \\ ($\mathrm{E_h}$)}}  & $M_\text{I}$ & $\aleph_\text{I}$, $\beth_\text{I}$ & \multicolumn{1}{c}{\makecell{CCSD corr. \\ error ($\mathrm{mE_h}$)}} \\
\hline
\multirow{1}{*}{Water} & 103.83 & 70 & 13  & 114.84 & 60 & 12 & 0.406  \\
\hline
\multirow{1}{*}{Ammonia}  & 139.40 & 80 & 11  &120.95& 90 & 10 & -0.042\\ \hline
 \multirow{1}{*}{P450 heme}  & 85.60  & 150 & 13 & 133.94 & 140 & 13 & -0.021\\\hline 
 \multirow{1}{*}{p-Benzyne}  & 727.33 & 430 & 13  & 621.68 & 430 & 15 & -0.370\\
\hline\hline
\end{tabular}
\caption{Block encoding details of the \emph{electron repulsion integral (ERI)} and the Hamiltonian used in for the quantum resource estimate in \cref{fig RE numeric}. 
This is an example of a two-body operator and so here both the Hamiltonian and observable are block encoded using the BLISS-THC technique. 
See \cref{table:kineticdata} for a description of the table headings. $\lambda_\text{I}$ is the one norm of the ERI in Hartree.}\label{table:eridata}
\end{table}

\end{appendix}
\end{document}